\documentclass[11pt,a4paper]{article}

\usepackage[utf8]{inputenc}
\usepackage[english]{babel}
\usepackage[T1]{fontenc}
\usepackage{amsmath}
\usepackage{amsfonts}
\usepackage{amssymb}
\usepackage{graphicx}
\usepackage{fullpage}
\usepackage{xcolor}
\usepackage{setspace}
\usepackage{amsthm}
\usepackage{algorithm}
\usepackage{algorithmic}
\usepackage{fancybox}
\usepackage{fancyhdr}
\usepackage{bbm}
\usepackage{multirow}
\usepackage{subfig}
\usepackage[numbers]{natbib}

\numberwithin{equation}{section}
\newcommand{\echelle}{0.35}
\newcommand{\echellePetit}{0.2}

\newcommand{\ind}[1]{\mathbbm{1}_{#1}}
\newcommand{\indd}{\mathbbm{1}}

\newcommand{\ensemble}[1]{ \{ #1 \} }

\newcommand{\R}{\mathbb{R}}

\newcommand{\E}{\mathbb{E}}
\newcommand{\V}{\mathbb{V}}
\newcommand{\C}{\mathbb{C}}
\newcommand{\Proba}{\mathbb{P}}

\newcommand{\Bcal}{\mathcal{B}}

\newcommand{\Dn}{\mathcal{D}_n}

\newcommand{\Dnbar}{\bar{\mathcal{D}}_n}

\theoremstyle{plain}
\newtheorem{prop}{Proposition}

\newtheorem*{remark}{Remark}

\theoremstyle{definition}

\theoremstyle{remark}

\newsavebox{\fmbox}

\title{Correlation and variable importance in random forests}
\author{Baptiste Gregorutti$^{12}$\footnote{Corresponding author: baptiste.gregorutti@safety-line.fr} , Bertrand Michel$^2$, Philippe Saint-Pierre$^2$}

\date{}

\begin{document}
\maketitle

\begin{center}
{\it
$^1$ Safety Line\\15 rue Jean-Baptiste Berlier, 75013 Paris, France\\
\bigskip
$^2$ Laboratoire de Statistique Th\'eorique et Appliqu\'ee\\
Universit\'e Pierre et Marie Curie\\
4 place Jussieu, 75252 Paris Cedex 05, France
}
\end{center}
\bigskip

\begin{abstract}
This paper is about variable selection with the random forests algorithm in presence of correlated predictors. In high-dimensional regression or classification frameworks, variable selection is a difficult task, that becomes even more challenging in the presence of highly correlated predictors. Firstly we provide a theoretical study of the permutation importance measure for an additive regression model. This allows us to describe how the correlation between predictors impacts the permutation importance. Our results motivate the use of the Recursive Feature Elimination (RFE) algorithm  for variable selection in this context. This algorithm recursively eliminates the variables using permutation importance measure as a ranking criterion. Next various simulation experiments illustrate the efficiency of the RFE algorithm for selecting a small number of variables together with a good prediction error. Finally, this selection algorithm is tested on the Landsat Satellite data from the UCI Machine Learning Repository.

\end{abstract}


\section{Introduction}

In large scale learning problems, in particular when the number of variables is much larger than the number of observations, not all of the variables are relevant for predicting the outcome of interest. Some irrelevant variables may have a negative effect on the model accuracy. Variable selection techniques, also called feature selection or subset selection, involve eliminating irrelevant variables and provide two main advantages. First, a model with a small number of variables is more interpretable. Secondly, the model accuracy might be improved and then avoid the risk of overfitting.

Many studies about variable selection have been conducted during the last decade. In \citet{svm:GE03}, the authors review three existing approaches: filter, embedded and wrapper. A filter algorithm, also known as variable ranking, orders the variables in a preprocessing step and the selection is done independently of the choice of the learning technique. Two classical ranking criteria are the Pearson correlation coefficient and the mutual information criterion as mentioned in the recent survey of \citet{fs:L+12}. The main drawback of this approach is that the choice of the selected variables is not induced by the performance of the learning method. The embedded approach selects the variables during the learning process. The two main examples are the Lasso (\citet{tib96}) for regression problems (the selection process is done through the $\ell_1$ regularization of the least square criterion) and decision trees (the selection is induced by the automatic selection of the splitting variables) such as CART algorithm (\citet{cart}).
A wrapper algorithm uses the learning process to identify an optimal set of variables among all possible subsets (see \citet{fs:KJ97}, \citet{fs:BL97}). The measure of optimality is usually defined by the error rate estimate. As it is impossible to evaluate all variable subsets when the dimension of the data is too large, the wrapper approach consists of using greedy strategies such as forward or backward algorithms. A heuristic is required to select the variables to be introduced or eliminated. This algorithm has been adapted for various contexts in the literature (see for instance \citet{svm:Guyon2002}, \citet{svm:R03}, \citet{rf:SLTW04}, \citet{rf:DA05}, \citet{fs:LS06}, \citet{rf:GPT10}).

A classical issue of variable selection methods is their instability: a small perturbation of the training  sample may completely change the set of selected variables. This instability is  a consequence of the data complexity in high dimensional settings (see \citet{fs:KPH07}, \citet{fs:KKH07}). In particular, the instability of variable selection methods increases when the predictors are highly correlated. For instance, \citet{lasso:B+12} have shown that the lasso tends to discard most of the correlated variables even if they are discriminants and randomly selects one representative among a set of correlated predictors. In the context of random forests, the impact of correlated predictors on variable selection methods has been highlighted by several simulation studies, see for instance \citet{rf:TL11}.  For real life applications it is of first importance to select a subset of variables which is the most stable as possible. One popular solution to answer the instability issue of variable selection methods consists in using bootstrap samples: a selection is done on several bootstrap subsets of the training data and a stable solution is obtained by aggregation of these selections. Such generic approach aims to improve both the stability and the accuracy of the method. This procedure is known as ``ensemble feature selection'' in the machine learning community. Several classification and regression techniques based on this approach have been developed (\citet{svm:B+03} in the context of Support vector regression, \citet{lasso:MB10} with the stability selection). \citet{fs:HGV11} provide a comparison of ensemble feature selections combined with several classification methods.

The random forests algorithm, introduced by \citet{rf:B01}, is a modification of bagging that aggregates a large collection of tree-based estimators.
This strategy has better estimation performances than a single random tree: each tree estimator has low bias but high variance whereas the aggregation achieves a bias-variance trade-off.  The random forests are very attractive for both classification and regression problems. Indeed, these methods have good predictive performances in practice, they work well for high dimensional problems and they can be used with multi-class output, categorical predictors and imbalanced problems. Moreover, the random forests provide some measures of the importance of the variables with respect to the prediction of the outcome variable.

Several studies have used the importance measures in variable selection algorithms (\citet{rf:SLTW04}, \citet{rf:DA05}, \citet{rf:GPT10}). The effect of the correlations on these measures has been studied in the last few years by \citet{rf:AK08}, \citet{rf:S+08}, \citet{rf:NM09}, \citet{rf:N+10}, \citet{rf:N11}, \citet{rf:AA11} and \citet{rf:TL11}.
However, there is no consensus on the interpretation of the importance measures when the predictors are correlated and more precisely there is no consensus on what is the effect of this correlation on the importance measures
(see e.g. \citet{rf:G09}, \citet{rf:N13}). One reason for this is that, as far as we know, no theoretical description of this effect has been proposed in the literature. This situation is particularly unsatisfactory as the importance measures are intensively used in practice for selecting the variables.

The contributions of this paper are two-fold. First, we give some theoretical descriptions of the effect of the correlations on the ranking of the variables produced by the permutation importance measure introduced in \citet{rf:B01}. More precisely, we consider a particular  additive regression model for which it is possible to express the permutation importance in function of the correlations between predictors. The results of this section are validated by a simulation study.

The second contribution of this paper is of algorithmic nature. We take advantage of the previous results to compare wrapper variable selection algorithms for random forests in the context of correlated predictors. Note that most of the variable selection procedures using the random forests are wrapper algorithms. It can be used also as a filter algorithm as in \citet{rf:HU13}. Two main wrapper algorithms are considered in the literature. These two rely on backward elimination strategies based on the ranking produced by the permutation importance measure. The first algorithm computes the permutation importance measures in the full model which produces a ranking of the variables. This ranking is kept unchanged by the algorithm (\citet{rf:SLTW04}, \citet{rf:DA05}, \citet{rf:GPT10}). The second algorithm was first proposed by \citet{svm:Guyon2002} in the context of support vector machines (SVM) and is referred to as Recursive Feature Elimination (RFE). This algorithm requires to update the ranking criterion at each step of a backward strategy: at each step the criterion is evaluated and the variable which minimizes this measure is eliminated. In the random forests setting, although less popular than the first one, this strategy has been implemented for instance in \citet{rf:J+04}. As far as we know, only one study by \citet{rf:SLTW04} compared the two approaches and concluded that the non recursive algorithm provides better results. However, their findings are based on a real life dataset without taking into account the effect of correlated predictors and are not confirmed by simulation studies. Moreover, this position goes against the results we find in Section~\ref{sec:theo}.

A simulation study has been performed to compare the performances of the recursive and the non recursive strategies. Several designs of correlated data encounters in the literature have been simulated for this purpose. As expected, the simulations indicate that the recursive algorithm provides better results.

The paper is organized as follows. We first introduce the statistical background of the random forests algorithm and the permutation importance measure. Section~\ref{sec:theo} provides some theoretical properties of this criterion in the special case of an additive regression model. Section~\ref{sec:algo} describes the RFE algorithm used for variable selection in a random forests analysis. Next, the effect of the correlations on the permutation importance and the good performances of the RFE algorithm in the case of correlated variables are emphasized in a simulation study. This algorithm is finally carried out to analyse satellite image from the Landsat Satellite data from the UCI Machine Learning Repository.

\section{Random forests and variable importance measures}
\label{sec:background}

Let us consider a variable of interest $Y$ and a vector of random variables $\mathbf X = (X_1, \ldots, X_p)$. In the regression setting a rule $\hat f$ for predicting $Y$ is a measurable function taking its values in $\R$. The prediction error of $\hat f$ is then defined by $R(\hat f) = \E \left[ ( \hat f(\mathbf{X}) - Y )^2  \right]$ and our goal is to estimate the conditional expectation $f(\mathbf x) = \E[Y | \mathbf X = \mathbf x]$.

Let $\Dn = \ensemble{(\mathbf X_1, Y_1), \ldots, (\mathbf X_n, Y_n)}$ be a learning set of $n$ i.i.d. replications of $(\mathbf{X},Y)$ where $\mathbf{X}_i = (X_{i1}, \ldots, X_{ip})$.   Since the true prediction error of  $\hat f$ is unknown in practice, we consider an estimator based on the observation of a validation sample $ \bar{\mathcal D}$:
$$\hat{R}(\hat f, \bar{\mathcal D}) = \dfrac{1}{|\bar{\mathcal D}|} \sum_{i: (\mathbf{X}_i, Y_i) \in \bar{\mathcal D}} (Y_i - \hat  f(\mathbf{X}_i) )^2.$$

Classification and regression trees, particularly CART algorithm due to \citet{cart}, are competitive techniques for estimating $f$. Nevertheless, these algorithms are known to be unstable insofar as a small perturbation of the training sample may change radically the predictions. For this reason, \citet{rf:B01} introduced the random forests as a substantial improvement of the decision trees. It consists in aggregating a collection of such random trees, in the same way as the bagging method also proposed by \citet{bagging}: the trees are built over $n_{tree}$ bootstrap samples $\Dn^1, \ldots, \Dn^{n_{tree}}$ of the training data $\Dn$. Instead of CART algorithm, a small number of variables is randomly chosen to determine the splitting rule at each node. Each tree is then fully grown or until each node is pure. The trees are not pruned. The resulting learning rule is the aggregation of all of the tree-based estimators denoted by $\hat{f}_1, \ldots, \hat{f}_{n_{tree}}$. The aggregation is based on the average of the predictions.

In parallel the random forests algorithm allows us to evaluate the relevance of a predictor thanks to variable importance measures. The original random forests algorithm computes three measures, the permutation importance, the z-score and the Gini importance. We focus here on the permutation importance due to \citet{rf:B01}. Broadly speaking, a variable $X_j$ can be considered as important for predicting $Y$ if by breaking the link between $X_j$ and $Y$ the prediction error increases. To break the link between $X_j$ and $Y$, Breiman proposes to randomly permute the observations of the $X_j$'s. It should be noted that the random permutations also breaks the link between $X_j$ and the other covariates. The empirical permutation importance measure can be formalized as follows: define a collection of out-of-bag samples $\ensemble{\Dnbar^t = \Dn \setminus \Dn^t, \, t=1, \ldots, n_{tree}}$ which contains the observations not selected in the bootstrap subsets. Let $\ensemble{\Dnbar^{tj}, \, t=1, \ldots, n_{tree}}$ denote the permuted out-of-bag samples by random permutations of the values of the $j$-th variable in each out-of-bag subset. The empirical permutation importance of the variable $X_j$ is defined by
\begin{equation}
\hat{I}(X_j) = \dfrac{1}{n_{tree}} \sum_{t=1}^{n_{tree}} \bigg[ \hat{R}(\hat{f}_t, \Dnbar^{tj}) - \hat{R} (\hat{f}_t, \Dnbar^{t}) \bigg].
\label{eq:vim}
\end{equation}
This quantity is the empirical counterpart of the permutation importance measure $I(X_j)$, as formalized recently in \citet{rf:Z+12}. Let $\mathbf X_{(j)} = (X_1, \ldots, X_j', \ldots, X_p)$ be the random vector such that $X_j'$ is an independent replication of $X_j$ which is also independent of $Y$ and of all of the others predictors, the permutation importance measure is given by
$$I(X_j) = \E \left[ \left( Y - f(\mathbf X_{(j)} ) \right)^2 \right] - \E \left[ \left( Y - f(\mathbf X )  \right)^2 \right].$$
The permutation of the values of $X_j$ in the definition of $\hat{I}(X_j)$ mimics the independence and the identical copy of the distribution of $X_j$ in the definition of $I(X_j)$.

While the permutation importance measure only depends on the joint distribution of $(\mathbf X,Y$), the empirical importance measure also strongly depends on the algorithm used to estimate the regression function. Consequently, the consistency of $\hat{I}(X_j)$ is relative to which particular algorithm has been chosen for estimating the regression function. 

Only a few results about the consistency of the empirical criterion \eqref{eq:vim} can be found in the literature.  The main contribution on this topic is due to \citet{rf:Z+12}. In this last paper, the algorithm of reinforcement learning trees is introduced and analysed.  Among other consistency results,  it is shown that $\hat{I}(X_j)$  converges to $ I(X_j)$  at an exponential rate,  in the particular case where the importance measures are computed with a purely random forests algorithm (see \citet{biau2008consistency}). Unfortunately, among other hypotheses, \citet{rf:Z+12} assume that the predictors are independent. Thus, this consistency result  can  be hardly  invoked in our context because our field of study is precisely  the effect of correlation on the  permutation importance. 

The consistency of Breiman's random forests algorithm has been shown very recently by~\citet{scornet2014consistency} in the context of the additive regression model \eqref{eq:modeladdit} introduced in the next section. However, this strong result does not  solve the question because the consistency of the random forests is not the same thing as the consistency of the permutation importance.  Moreover, as in \citet{rf:Z+12}, it is assumed by \citet{scornet2014consistency} that the predictors are independent.  As far as we know, the consistency of the empirical  permutation importance, when computed in the context of the Breiman's random forests, is still an open problem.

The permutation importance measure can be used to rank or select the predictors. Among others criteria, the permutation importance measure has shown good performances for leading variable selection algorithms. Nevertheless variable selection is a difficult issue especially when the predictors are highly correlated. In the next section we investigate deeper the properties of the permutation importance measure in order to understand better how this quantity depends on the correlation between the predictors.

\section{Permutation importance measure of correlated variables}
\label{sec:theo}

Previous results about the impact of correlation on the importance measures are mostly based on experimental considerations.
We give a non exhaustive review of these contributions and we compare them with our theoretical results.
\citet{rf:AK08} observe that the Gini measure is less able to detect the most relevant variables when the correlation increases and they mention that the same is true for the permutation importance. The experiments of \citet{rf:AA11} confirm these observations. \citet{rf:GPT10} study the sensitivity of the empirical permutation importance measure to many parameters, in particular they study the sensitivity to the number of correlated variables. Recently, \citet{rf:TL11} identify what they call the ``correlation bias''. Note that it does not correspond to a statistical bias. More precisely, these authors observe two key effects of the correlation on the permutation importance measure: first, the importance values of the most discriminant correlated variables are not necessarily higher than a less discriminant one, and secondly the permutation importance measure depends on the size of the correlated groups.

Since previous studies are mainly based on numerical experiments, there is obviously a need to provide theoretical validations of these observations. We propose below a first theoretical analysis of this issue, in a particular statistical framework. In the rest of the section, we assume that the random vector $(\mathbf X, Y)$ satisfies the following additive regression model:
\begin{equation}
\label{eq:modeladdit}
 Y=\sum_{j=1}^p f_j(X_j) + \varepsilon,
\end{equation}
where $\varepsilon$ is such that $\E[\varepsilon|\mathbf X]=0$, $\E[\varepsilon^2|\mathbf X]$ is finite and the $f_j$'s are measurable functions.  In other words, we require that the regression function can be decomposed into $f(\mathbf x) = \sum_{j=1}^p f_j(x_j)$. In the sequel, $\V$ and $\C$ denote variance and covariance.
\begin{prop}
\label{prop:ResAdditif}
\begin{enumerate}
\item Under model \eqref{eq:modeladdit}, for any $j \in \{1,\dots,p\}$, the permutation importance measure satisfies
$$ I(X_j) = 2\V[f_j(X_j) ] . $$
\item Assume moreover that for some $j \in \{1,\dots,p\}$ the variable $f_j(X_j) $ is centered. Then,
$$I(X_j)  = 2 \C[Y, f_j(X_j)]  - 2  \sum_{k\neq j} \C[f_j(X_j), f_k(X_k)].$$
\end{enumerate}
\end{prop}

\begin{proof}
see Appendix \ref{sec:proofResAdditif}
\end{proof}

In this framework, the permutation importance corresponds to the variance of $f_j(X_j) $, up to a factor 2. The second result of Proposition~\ref{prop:ResAdditif} is the key point to study the influence of the correlation on the permutation measure. This result strongly depends on the additive structure of the regression function $f$ and it seems difficult to give such a simple expression of the permutation importance without assuming this additive form for the regression function. 

If $(\mathbf X, Y)$ is assumed to be a normal vector it is possible to specify the permutation importance measure. Note that in this context the conditional distribution of $Y$ over $\mathbf X$ is also normal and the conditional mean $f$ is a linear function: $f(\mathbf x) = \sum_{j=1}^p \alpha_j x_j$ with $\alpha = (\alpha_1, \ldots, \alpha_p)^t$ a sequence of deterministic coefficients (see for instance \citet{rao73}, p. 522).

\begin{prop}
\label{prop:Gauss}
Consider a Gaussian random  vector $(\mathbf X, Y) \sim \mathcal{N}_{p+1} \bigg(0,
		\begin{pmatrix}
		  	C & \boldsymbol{\tau} \\
			\boldsymbol{\tau}^t & \sigma_y^2
		\end{pmatrix} \bigg)$, where $\boldsymbol{\tau} = (\tau_1, \dots,\tau_p)^t$ with  $\tau_j = \C(X_j, Y) $, $\sigma_y^2 >0$ and $C = [\C(X_j, X_k)]$ is the  non degenerated variance-covariance matrix of $\mathbf X$. Then, for any $j \in \{1,\dots,p\}$,
$$I(X_j) = 2 \alpha_j^2 \V(X_j) = 2 \alpha_j \C(X_j, Y) - 2 \alpha_j \sum_{k\neq j} \alpha_k \C(X_j, X_k),$$
where $\alpha_j = [C^{-1} \boldsymbol{\tau}]_j.$
\end{prop}

\begin{proof}
see Appendix \ref{sec:proofGauss}
\end{proof}

Note that if we consider a linear function $f: \R^p \mapsto \R$, a random vector $\mathbf X$ of $\R^p$ and a random variable $\varepsilon$ such that $(\mathbf X, \varepsilon)$ is a multivariate normal vector, and if we define the outcome by $Y = f(\mathbf X) + \varepsilon$, then the vector $(\mathbf X, Y)$ is clearly a multivariate normal vector. Thus, the assumption on the joint distribution of $(\mathbf X, Y)$ is in fact mild in the regression framework. Note also that Proposition~\ref{prop:Gauss}  corresponds to the criterion used in \citet{svm:Guyon2002} as a ranking criterion in the SVM-RFE algorithm with $\V(X_j) = 1$.

\medskip

We now discuss the effect of the correlation between predictors on the importance measure by considering Gaussian regression models with various configurations of correlation between predictors:
\paragraph{Case 1: Two correlated variables.} Consider the simple context where $(X_1, X_2, Y) \sim \mathcal{N}_{3} \bigg(0,
	\begin{pmatrix}
	  	C & \boldsymbol{\tau} \\
		\boldsymbol{\tau}^t & 1
	\end{pmatrix}\bigg)$ with
$$C =
	\begin{pmatrix}
	  	1 & c\\
		c & 1
	\end{pmatrix},$$
and $\boldsymbol{\tau}^t = (\tau_0, \tau_0)$ with $\tau_0 \in (-1;1)$. Since the two predictors have the same correlation $\tau_0$ with the outcome, and according to Proposition~\ref{prop:Gauss}, we have for $j \in \ensemble{1,2}$:
		\begin{eqnarray}
		\alpha_j &=& \dfrac{\tau_0}{1+c}.\label{eq:Gauss:case:p2}
		\end{eqnarray}

\noindent Consequently, the permutation importance for both $X_1$ and $X_2$ is
		\begin{equation} \label{ImpCase1}
		I(X_j) = 2 \bigg( \dfrac{\tau_0}{1+c} \bigg)^2, \, j \in \ensemble{1,2}.
		\end{equation}
For positive correlations $c$, the importance of the two variables $X_1$ and $X_2$ decreases when $c$ increases. This result is quite intuitive: when one of the two correlated variables is permuted, the error does not increase that much because of the presence of the other variable, which carries a similar information. The value of the  prediction error after permutation is then close to the value of the prediction error without permutation and the importance is small.

\paragraph{Case 2: Two correlated and one independent variables.}
We add to the previous case an additional variable $X_3$ which is assumed to be independent of $X_1$ and $X_2$, $X_1$ and $X_2$ being unchanged. It corresponds to
$$C = \begin{pmatrix}
		  	1 & c & 0 \\
			c & 1 & 0 \\
			0 & 0 & 1
\end{pmatrix},$$
and $\boldsymbol \tau^t = (\tau_0, \tau_0, \tau_3)$. It can be easily checked that
$$C^{-1} (\tau_0, \tau_0, \tau_3)^t = \bigg( \dfrac{\tau_0}{1+c}, \dfrac{\tau_0}{1+c}, \tau_3 \bigg)^t.$$

\noindent Thus, \eqref{eq:Gauss:case:p2} and  \eqref{ImpCase1} still hold and   $\alpha_3 = \tau_3$. As a consequence, $I(X_3) = 2 \tau_3^2$ can be larger than $I(X_1)$ and $I(X_2)$ if the  correlation $c$ is sufficiency large even if $\tau_0 > \tau_3$. This phenomenon corresponds to the observation made by \citet{rf:TL11}.

\paragraph{Case 3: $p$ correlated variables.} We now consider $p$ correlated variables where
$$C = \begin{pmatrix}
		  	1 		& c 		& \cdots 	& c \\
			c 		& 1 		& \cdots	& c \\
			\vdots 	& \vdots	& \ddots	& \vdots \\
			c 		& c 		& \cdots	& 1
\end{pmatrix},$$
and $\boldsymbol \tau^t = (\tau_0, \dots, \tau_0)$.
In this context the $\alpha_j$'s are equal. The results of Case 1 can be generalized to that situation as shown by the following result:
\begin{prop} \label{lem:invC}
Assume that the correlation matrix $C$ can be decomposed into $C = (1-c) I_p + c \indd \indd^t$, where $I_p$ is the identity matrix and $\indd = (1, \ldots, 1)^t$. Let $\boldsymbol \tau = (\tau_0, \ldots, \tau_0)^t  \in \R^p$. Then for all $j \in \ensemble{1, \ldots, p}$: $$[C^{-1} \boldsymbol \tau]_j = \dfrac{\tau_0}{1-c+pc},$$ and consequently
\begin{equation}
I(X_j) = 2 \left( \dfrac{\tau_0}{1-c+pc} \right)^2.
\label{eq:Gauss:case:p}
\end{equation}
\end{prop}

\begin{proof}
see Appendix \ref{sec:proofinvC}
\end{proof}

\noindent The proposition shows that the higher the number of correlated variables is, the faster the permutation importance of the variables decreases to zero (see Fig.~\ref{fig:case3}). It confirms the second key observation of \citet{rf:TL11}.
\begin{figure}[ht]
	\begin{center}
	\includegraphics[scale=\echelle]{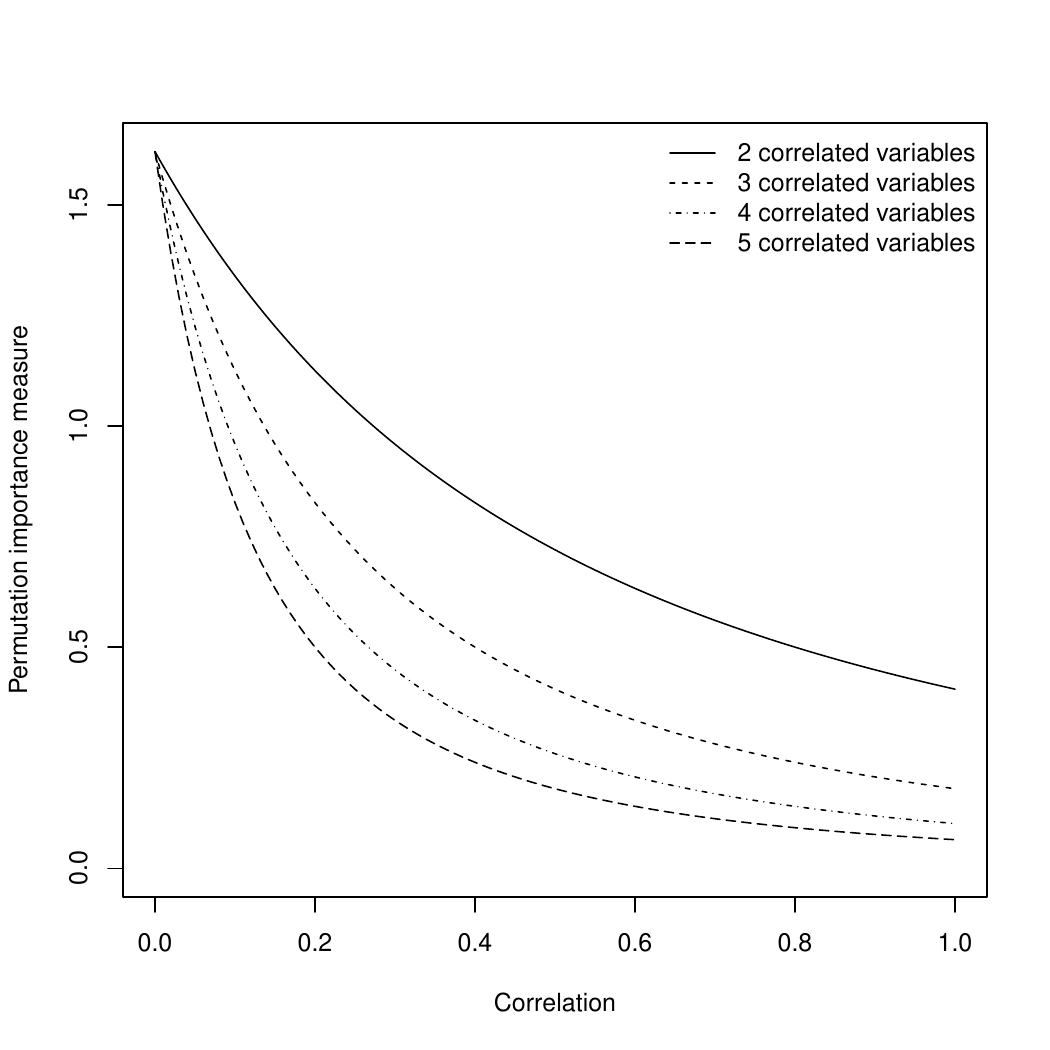}
	\caption{Case 3 - Permutation importance measure \eqref{eq:Gauss:case:p} versus the predictor correlation for $p \in \ensemble{2,3,4,5}$ and for $\tau_0=0.9$.}
	\label{fig:case3}
	\end{center}
\end{figure}
\paragraph{Case 4: One group of correlated variables and one group of independent variables.} Let us now assume that two independent blocks of predictors are observed. The first block corresponds to the $p$ correlated variables $X_1, \ldots, X_p$ of Case 3 whereas the second block is composed of $q$ independent variables $X_{p+1}, \ldots, X_{p+q}$. This case is thus a generalization of Case 2 where

$$C = \begin{pmatrix}
		  	1 		& \cdots 	& c 		& 0 		& \cdots 	& 0\\
			\vdots 	& \ddots	& \vdots 	& \vdots 	& 			& \vdots \\
			c 		& \cdots	& 1			& 0			& \cdots	& 0\\
			0 		& \cdots	& 0			& 1			& \cdots	& 0\\
			\vdots 	& 			& \vdots 	& \vdots 	& \ddots	& \vdots \\
		  	0 		& \cdots 	& 0 		& 0 		& \cdots 	& 1
\end{pmatrix},$$
and $\boldsymbol \tau^t = (\tau_0, \dots, \tau_0, \tau_{p+1}, \dots, \tau_{p+q})$.
It can be checked that the importance of the correlated variables is still given by \eqref{eq:Gauss:case:p} and that $I(X_j) = 2 \tau_j^2$ for $j \in \ensemble{p+1, \ldots, p+q}$. Again, the independent variables may show higher importance values even if they are less informative than the correlated ones.

\paragraph{Case 5: Anti-correlation.} All the previous cases consider positive correlations between the predictors. We now look at the effect of anti-correlation on the permutation measure. Let us consider two predictors $X_1$ and $X_2$ such that $X_2= - \rho X_1 + \varepsilon $ where $X_1$ and $\varepsilon$ are independent and $\rho \in (0,1]$. The correlation between  $X_1$ and $X_2$ equals $-\rho$, assuming that the variances of $X_1$ and $X_2$ are equal to 1. The permutation importance increases when $\rho$ grows to 1 according to Equation \eqref{eq:Gauss:case:p2}. This surprising phenomenon can be explained intuitively: if $\rho$ is close to -1, we need both $X_1$ and $X_2$ in the model to explain $Y$ because they vary in two opposite directions. Consequently, the random permutations of one of the two variables induces a high prediction error. Finally, the permutation importance of these two variables is high for $\rho$ close to -1.

\medskip

\subsection*{Permutation importance measure for classification}

 We close this section by giving few elements regarding permutation importance measures in the classification framework. In this context, $Y$ takes its values in $\ensemble{0,1}$. The error of a rule $f$ for predicting $Y$ is $R(f) = \mathbb{P} \left[ f(\mathbf{X}) \neq Y \right] $. The function minimizing $R$ is the Bayes classifier defined by $f(\mathbf x) = \ind{\eta(\mathbf x) > 0.5}$, where $\eta(\mathbf x)= \E[Y|\mathbf X=\mathbf x] = \Proba[Y=1|\mathbf X=\mathbf x]$ is the regression function. Given a classification rule $\hat f$, we consider its empirical error based on the learning set $\Dn$ and a validation sample $\bar{\mathcal D}$:
$$\hat{R}(\hat f,\bar{\mathcal D}) = \dfrac{1}{|\bar{\mathcal D}|} \sum_{i: (\mathbf{X}_i, Y_i) \in \bar{\mathcal D}} \ind{ \hat f (\mathbf{X}_i) \neq Y_i}\nonumber.$$
Consequently, the permutation importance measure is
$$I(X_j)  =  \Proba \left[ Y \neq f(\mathbf X_{(j)} ) \right] - \Proba \left[Y \neq f(\mathbf X ) \right],$$
and its empirical counterpart is
$$\hat{I}(X_j) = \dfrac{1}{n_{tree}} \sum_{t=1}^{n_{tree}} \bigg[ \hat{R}(\hat{f}_t, \Dnbar^{tj}) - \hat{R} (\hat{f}_t, \Dnbar^{t}) \bigg],$$
as in Equation \eqref{eq:vim}. We can equivalently rewrite the importance $I(X_j)$ as
\begin{eqnarray}
I(X_j) &=& \E[(Y - f(\mathbf X_{(j)}))^2] - \E[(Y - f(\mathbf X))^2]\notag  \\
&=& \E[(Y - \eta(\mathbf X_{(j)}))^2] - \E[(Y - \eta(\mathbf X))^2] \label{eq:ImportanceClassif}
\end{eqnarray}
Of course the regression function does not satisfy the additive model \eqref{eq:modeladdit} but we can consider alternatively the additive logistic regression model:
$$\mbox{Logit}(\eta(\mathbf x)) = \sum_{j=1}^p f_j(X_j).$$
However, the permutation importance measure \eqref{eq:ImportanceClassif} cannot be easily related to the variance terms $\V[f_j(X_j)]$. In fact, this is possible by defining a permutation importance measure  $\tilde I  $ on the odd ratios $\frac{\eta(\mathbf x)}{1-\eta(\mathbf x)}$ rather than on the regression function as follows:
$$ \tilde I ( X_j) = \E \bigg[ \bigg(\log \frac{\eta(\mathbf X)}{1-\eta(\mathbf X)} - \log \frac{\eta(\mathbf X_{(j)})}{1-\eta(\mathbf X_{(j)})}\bigg)^2\bigg].$$
Indeed, straightforward calculations show that
$$ \tilde I ( X_j) = 2  \V[f_j(X_j)].$$
Roughly,  the permutation of $X_j$ has an impact in $I$ only if the permutations change the predicted class (for instance when the odd ratios are close to 1). In contrast the perturbation of the odd ratio due to a permutation of $X_j $ in $\tilde I$  is taken into account in $\tilde I$ whatever the value of the odd ratio. Nevertheless, the calculations we propose for the regression framework can be hardly adapted in this context, essentially because $\tilde I $ cannot be easily expressed in function of the correlations between variables. Moreover, as explained before, $\tilde I$ is less relevant than $I$ for the classification purpose.

\medskip

\medskip

The results of this section show that the permutation importance is strongly sensitive to the correlation between the predictors. Our results also suggest that, for backward elimination strategies, the permutation importance measure should be recomputed each time a variable is eliminated. We study this question in the next section.

\section{Wrapper algorithms for variable selection based on importance measures}
\label{sec:algo}

In this section we study wrapper variable selection algorithms with random forests in the context of highly correlated predictors. In the applications we have in mind, the number of predictors is large and it is then impossible to evaluate the error of all the subsets of variables. Such an exhaustive exploration is indeed ruled out by the computational cost. One solution to this issue, which has been investigated in previous studies, is to first rank the variables according to some criterion and then to explore the subsets of variables according to this ranking. Several papers follow this strategy, they differ to each other first on the way the error is computed and second on the way the permutation importance of variables is updating during the algorithm.

Choosing the error estimator is out of the scope of this paper although various methods are proposed in the literature on this issue. For instance, \citet{rf:DA05} and \citet{rf:GPT10} use the out-of-bag (OOB) error estimate whereas \citet{rf:J+04} use both the OOB and a validation set to estimate the error. Finally, in order to avoid the selection bias described in \citet{fs:AM02}, \citet{rf:SLTW04} use an external 5-fold cross-validation procedure: they produce several variable selections on the 5 subsets of the training data and compute the averaged CV errors. In the sequel, the algorithms are performed by computing two kinds of errors : (i) OOB error which is widely used but often too optimistic as discussed in \citet{rf:B01}, (ii) validation set error which is more suitable but can not be considered in all practical situations.

We focus on the way the permutation importance measure is used in the algorithms. The first approach consists in computing the permutation importance only at the initialization of the algorithm and then to follow a backward strategy according to this ``static'' ranking. The method can be summarized as follows:
\begin{enumerate}
\item Rank the variables using the permutation importance measure
\item Train a random forests \label{algo1:step:train}
\item Eliminate the less relevant variable(s) \label{algo1:step:elim}
\item Repeat steps \ref{algo1:step:train} and \ref{algo1:step:elim} until no further variables remain
\end{enumerate}

This strategy is called Non Recursive Feature Elimination (NRFE). \citet{rf:SLTW04}, \citet{rf:DA05} have developed such backward algorithms. More elaborated algorithms based on NRFE has been proposed in the literature as in \citet{rf:GPT10}. Since we are interested here in the effect of updating the measure importance, we only consider here the original version of NRFE.

The second approach called Recursive Feature Elimination (RFE) is inspired by \citet{svm:Guyon2002} for SVM. It requires an updating of the permutation importance measures at each step of the algorithm. This strategy has been implemented in \citet{rf:J+04}. The RFE algorithm implemented in this paper can be summarized as follows:
\begin{enumerate}
\item Train a random forests \label{algo2:step:train}
\item Compute the permutation importance measure
\item Eliminate the less relevant variable(s) \label{algo2:step:elim}
\item Repeat steps \ref{algo2:step:train} to \ref{algo2:step:elim} until no further variables remain
\end{enumerate}
The two approaches are compared in \citet{rf:SLTW04}. The authors find that NRFE has better performance than RFE algorithm for their real life application. But as far as we know, no simulation studies have been carried out in the literature to confirm their observations. Moreover, this position goes against the theoretical considerations detailed above.

The results of the previous section show that the permutation importance measure of a given variable strongly depends on its correlation with the other ones and thus on the set of variables not yet eliminated in the backward process.
As a consequence, RFE algorithm might be more reliable than NRFE since the ranking by the permutation importance measure is likely to change at each step. In the end, RFE algorithm can select smaller size models than NRFE since the most informative variables are well ranked in the last steps of the backward procedure even if they are correlated.
In addition, by recomputing the permutation importance measure, we make sure that the ranking of the variables is consistent with their use in the current forest.

Let us consider a simple example to illustrate these ideas: we observe two correlated variables highly correlated with the outcome, four independent variables less correlated to the outcome and six irrelevant variables. More precisely, the variables are generated under the assumptions of Proposition~\ref{prop:Gauss}: the correlation between the two relevant variables and the outcome is set to 0.7, the correlation between these variables is 0.9. The correlation between the independent variables and the outcome is 0.6. In addition, the variance of the outcome is set to 2 in order to have a positive-definite covariance matrix in the normal multivariate distribution. Figure~\ref{fig:stepbystep} represents the boxplots of the permutation importance measures at several steps of the RFE algorithm.
At the beginning of the algorithm, the permutation importance measure of the two first variables is lower than the independent ones (V3 to V6) even if they are more correlated to the outcome. Regarding the prediction performances of the selection procedure, one would like to select firstly one of the most informative variables (V1 or V2 in our example).
At the last steps of the algorithm (Fig.~\ref{fig:step5}), one of the most relevant variables is eliminated and there is no more correlations between the remaining variables. The permutation importance of variable V1 becomes larger than the other variables. Consequently, RFE algorithm firstly selects this variable whereas it is selected in fifth position when using NRFE according to the ranking shown in Figure~\ref{fig:step1}.

\begin{figure}[ht]
	\begin{center}
	\subfloat[Step 1]{\includegraphics[scale=\echellePetit]{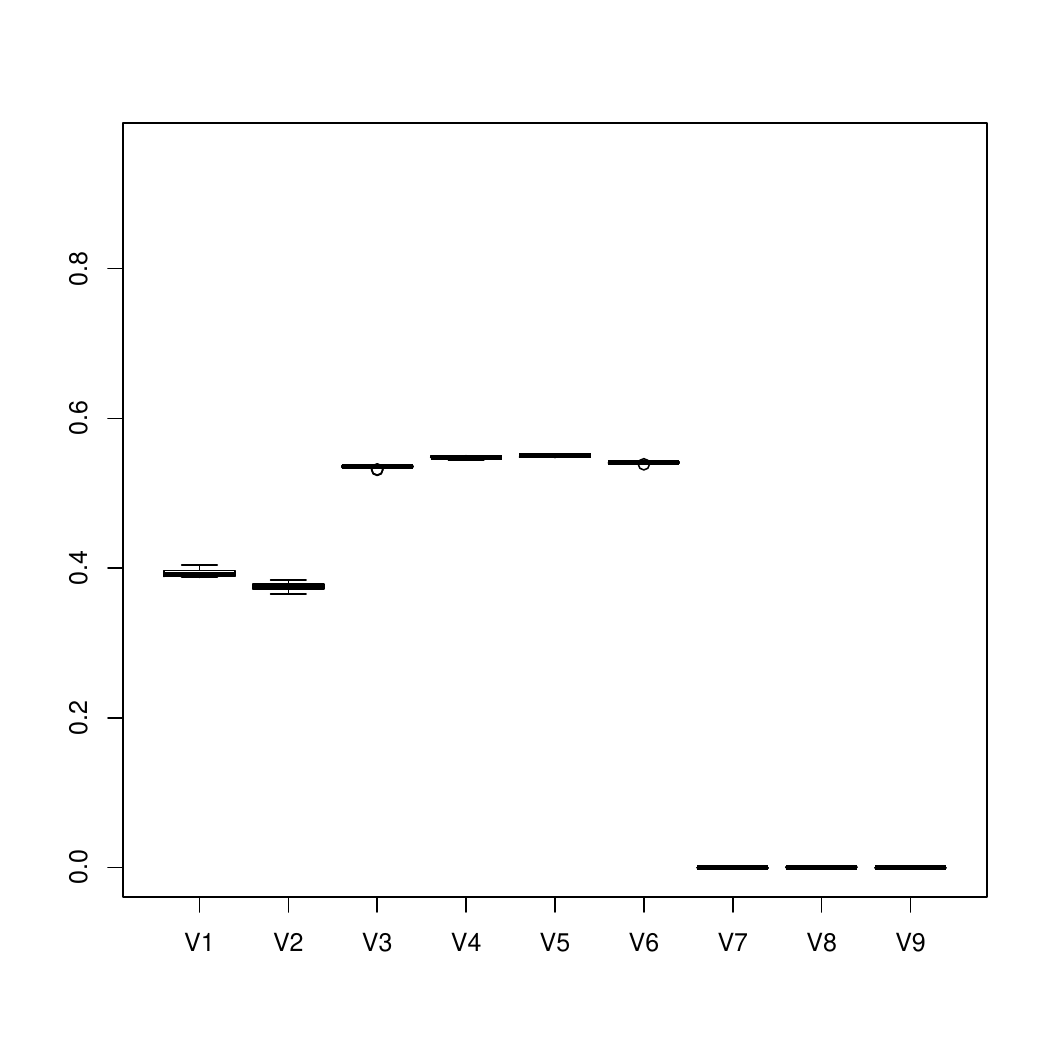}\label{fig:step1}}
	\subfloat[Step 4]{\includegraphics[scale=\echellePetit]{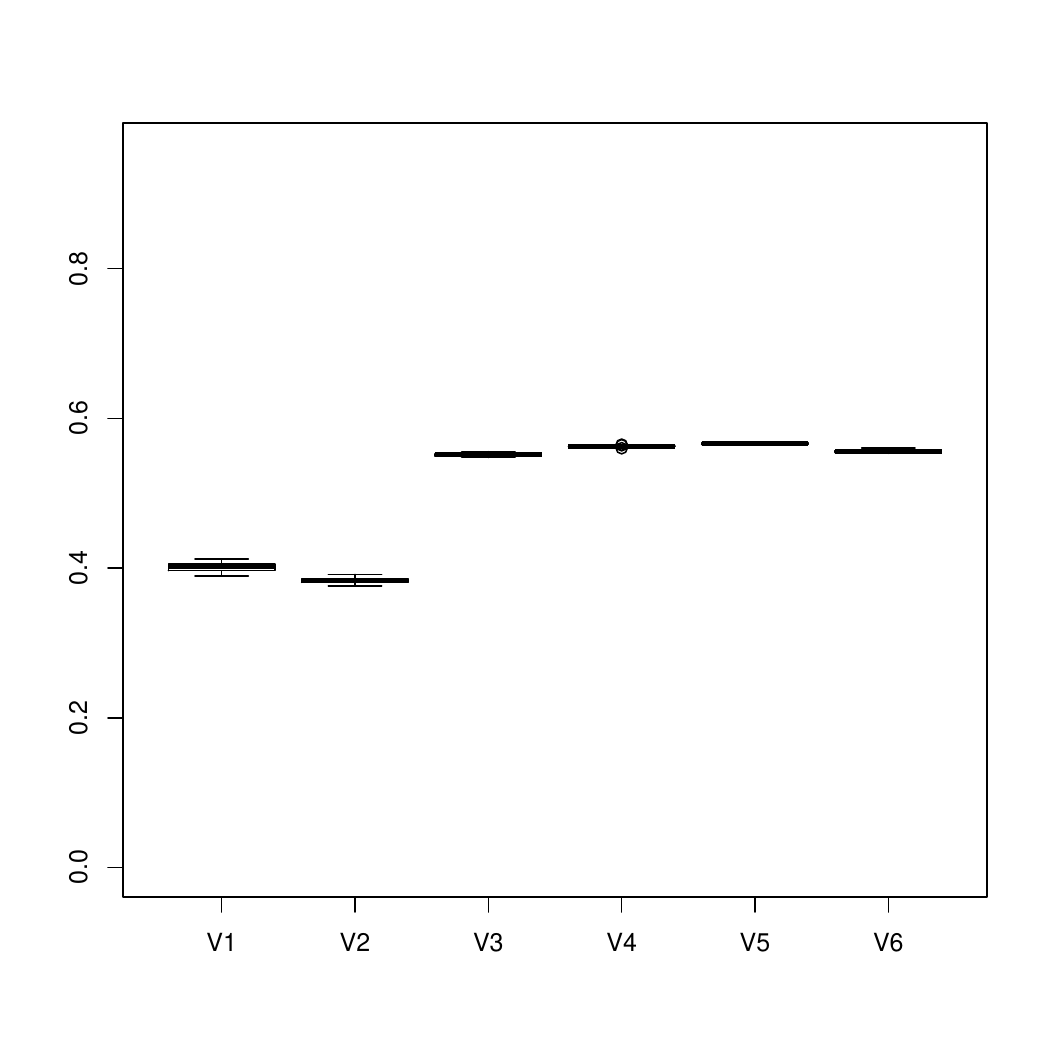}}
	\subfloat[Step 5]{\includegraphics[scale=\echellePetit]{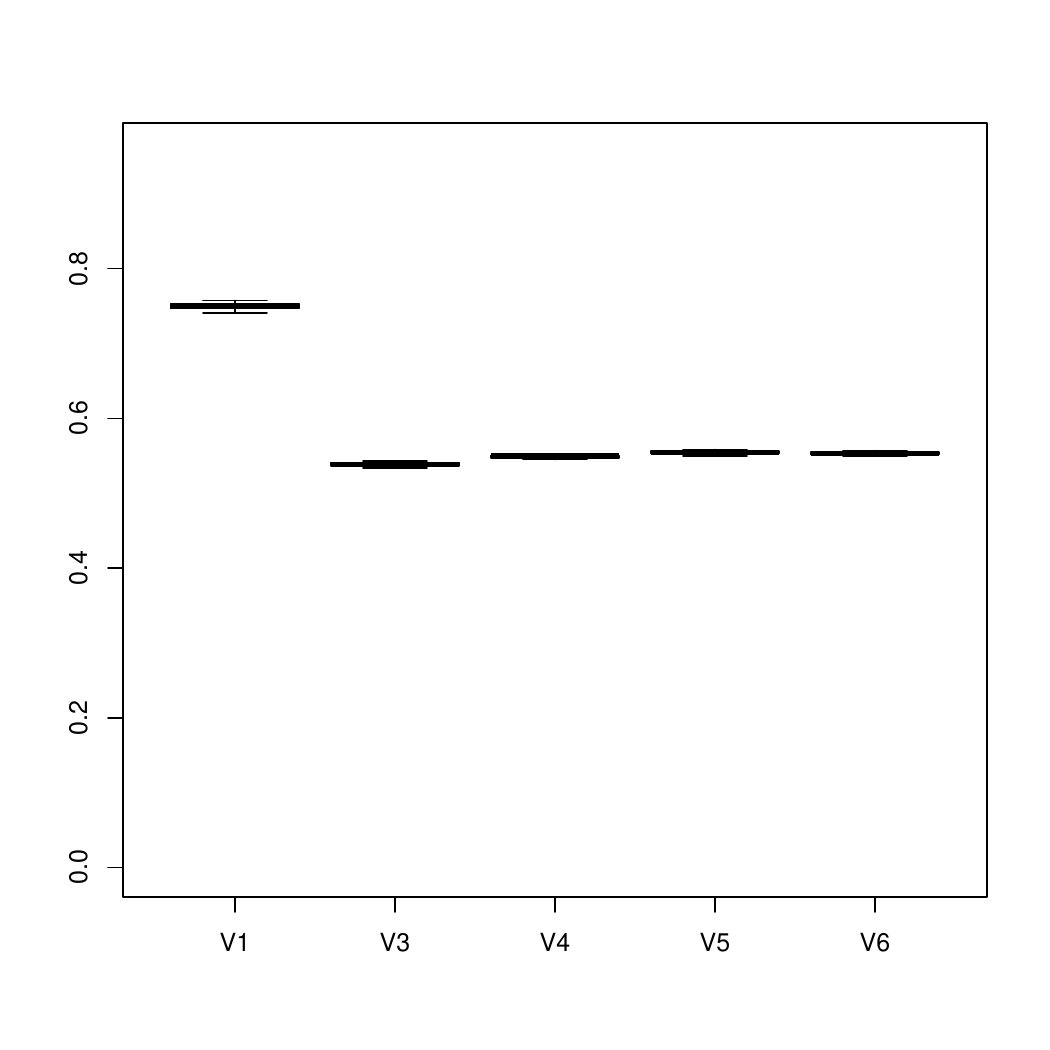}\label{fig:step5}}
	\caption{RFE algorithm step by step with six relevant variables (two correlated and four independent) and three irrelevant variables.}
	\label{fig:stepbystep}
	\end{center}
\end{figure}

In real life applications we usually need to find small size models with good performances in prediction. Thus, it is of first importance to efficiently reduce the effect of the correlations at the end of the backward procedure. By recomputing the variable importances at each step of the algorithm, the RFE algorithm manages to find small models which are efficient in term of prediction.

\begin{remark}
Note that another importance measure has been proposed in \citet{rf:S+08} for variable ranking with random forests in the context of correlated predictors. This importance measure, called conditional importance measure, consists in permuting variables conditionally to correlated ones. This method shows good performances for a small number of predictors. However it is computationally demanding and consequently it can be hardly implemented for problems with several hundreds of predictors.
\end{remark}

\section{Numerical experiments}
\label{sec:num}

In this section, we verify with several experiments that the results proved in Section~\ref{sec:theo} for the permutation importance measure are also valid for its empirical version \eqref{eq:vim}. RFE and  NRFE approaches  are compared for both classification and regression problems.

In the experiments, the number of trees in a  forests is set to $n_{tree}=1000$ and the number of variables randomly chosen for each split of the trees is set to the default value $m_{try} = \sqrt{p}$. For the NRFE algorithm, the permutation importance measure is averaged over 20 iterations as a preliminary ranking of the variables.

\subsection{Correlation effect on the empirical permutation importance measure}
\label{sec:corrEmp}

\begin{figure}
	\begin{center}
	\subfloat[]{\includegraphics[width=0.33\textwidth]{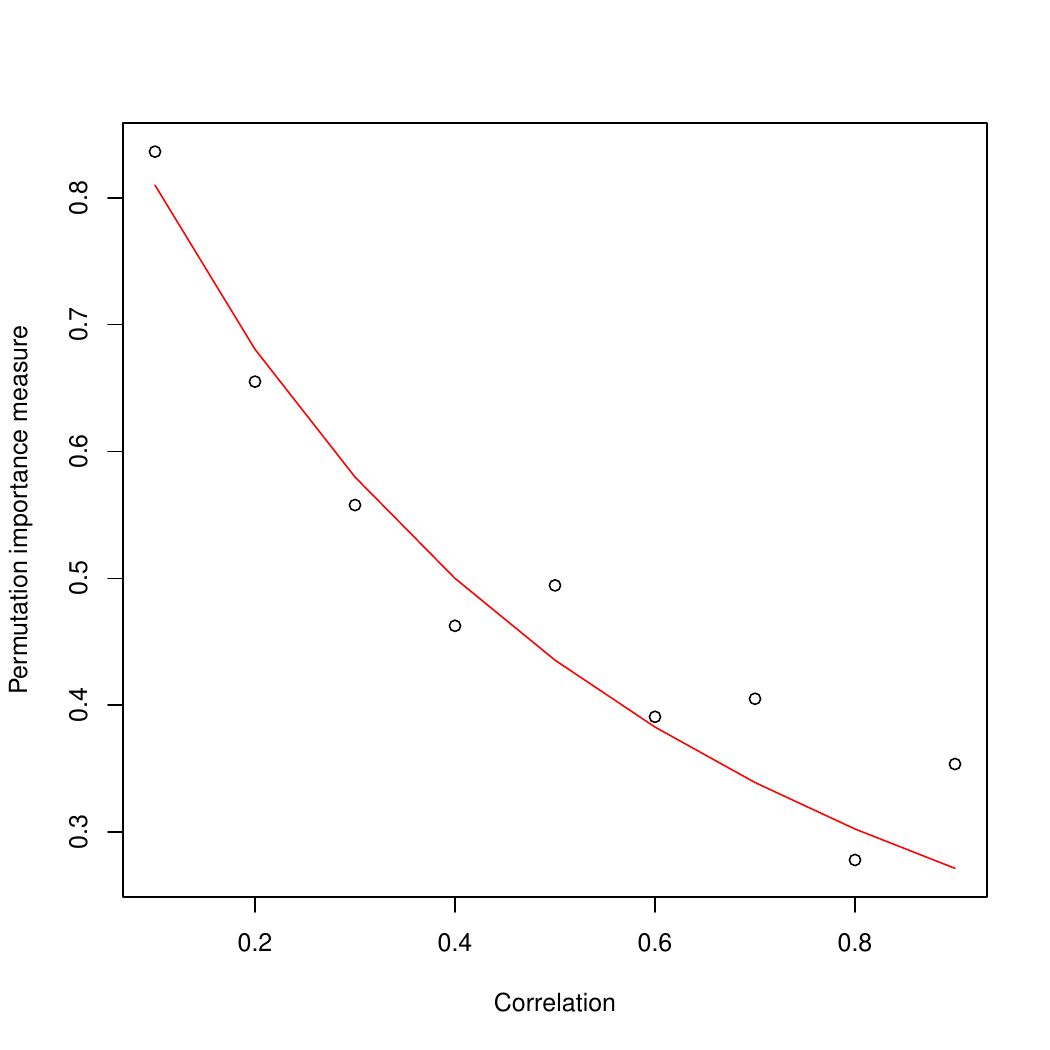}\label{fig:gauss1}}
	\subfloat[]{\includegraphics[width=0.33\textwidth]{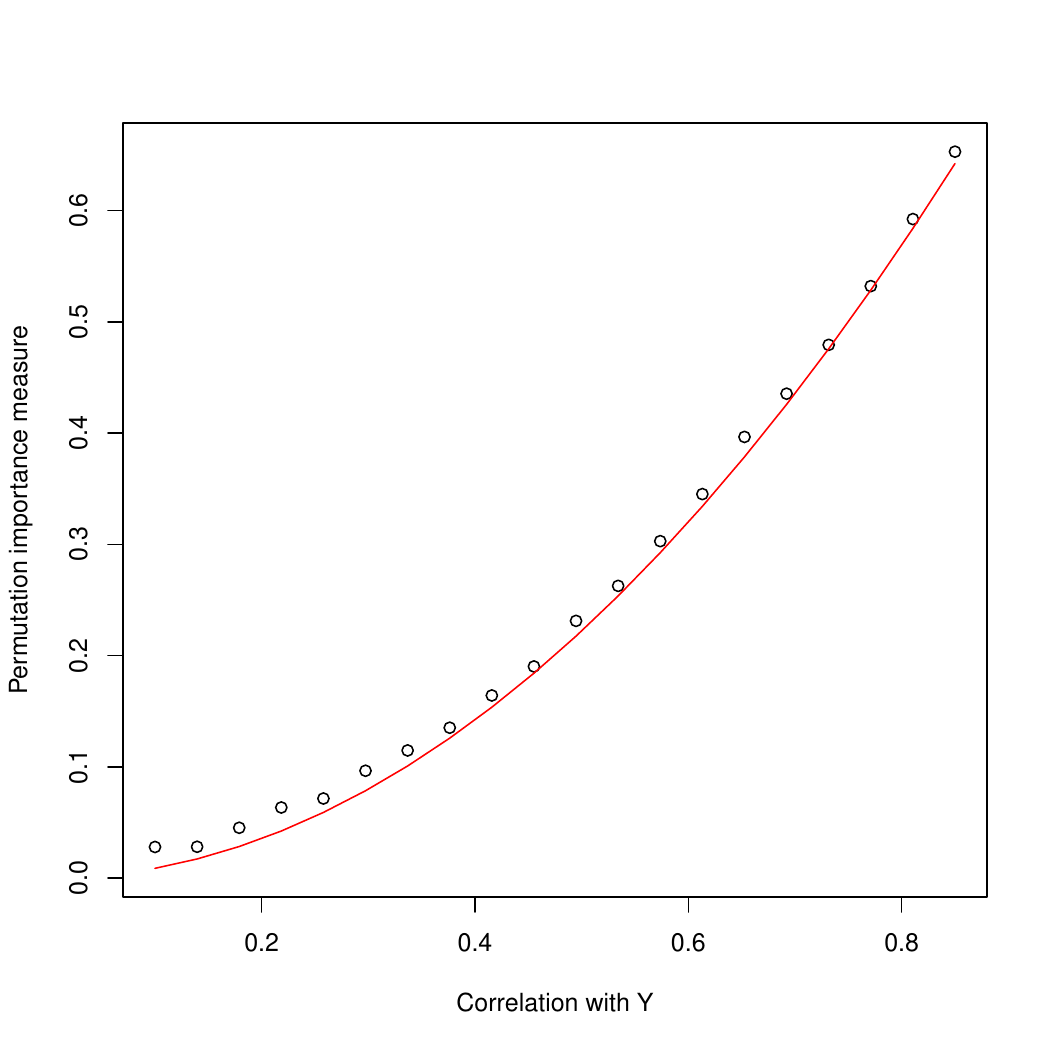}\label{fig:gauss3}}
	\subfloat[]{\includegraphics[width=0.33\textwidth]{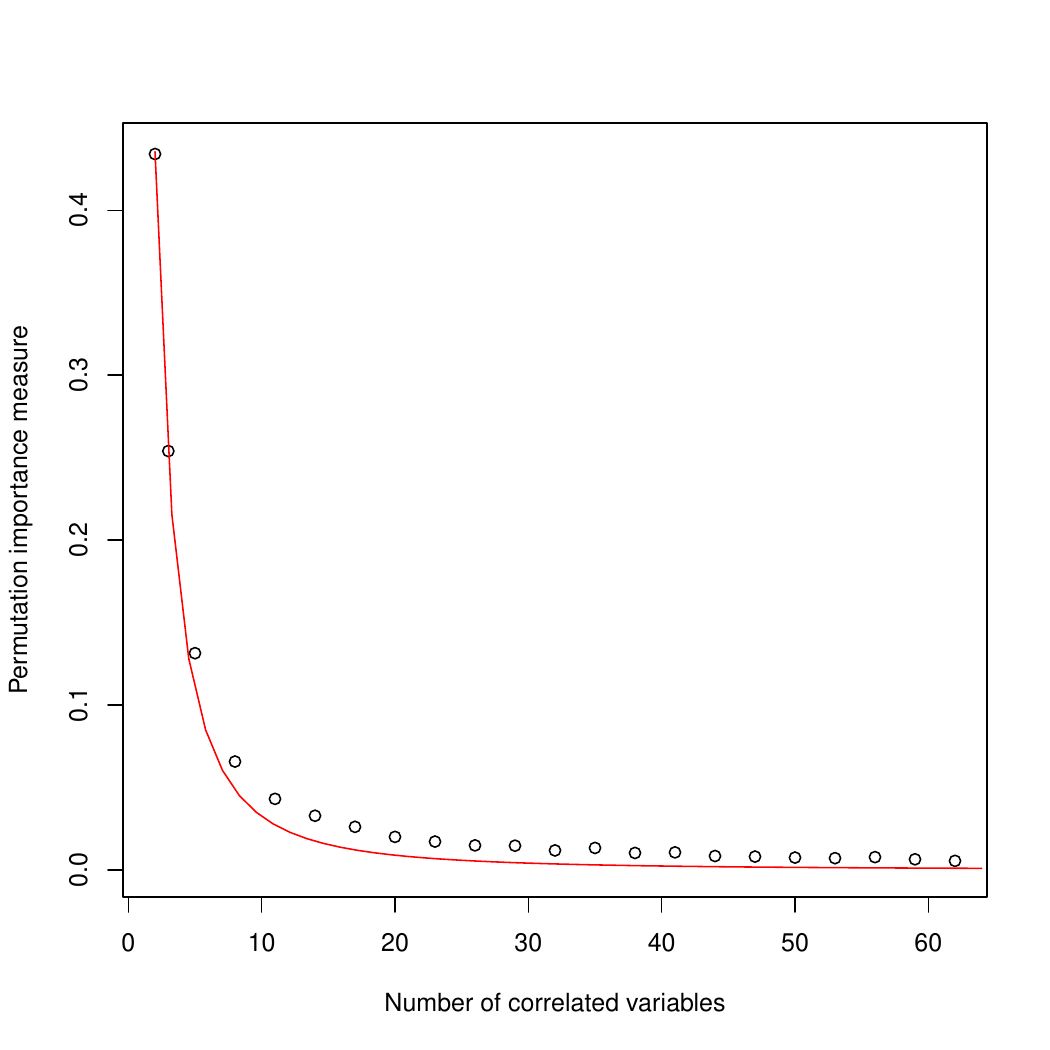}\label{fig:gauss2}}
	\caption{Permutation importance measure versus the correlation (left), the correlation with $Y$ (center) and the number of correlated predictors (right). The curves come from the expression of the permutation importance given in  Proposition~\ref{lem:invC}.}
	\label{fig:gauss}
	\end{center}
\end{figure}

This experiment is carried out under the assumptions of Proposition~\ref{prop:Gauss} and more precisely it corresponds to the regression problem presented in Case 1 and Case 3. Using the notations introduced in Section~\ref{sec:theo}, the variance-covariance matrix $C$ of the $p$ covariates $X_1, \ldots, X_p$ has the form
$$C = \begin{pmatrix}
		  	1 		& c 		& \cdots 	& c \\
			c 		& 1 		& \cdots	& c \\
			\vdots 	& \vdots	& \ddots	& \vdots \\
			c 		& c 		& \cdots	& 1
\end{pmatrix},$$ where $c=\C(X_j, X_k)$, for $j \neq k$. The correlation between the $X_j$'s and $Y$ is denoted by $\tau_0$. Two situations are considered. First, we take $p=2$ (Case 1). The permutation importance measure is given by Proposition \ref{lem:invC}:
\begin{equation} \label{eq:IX1}
 I(X_1) = I(X_2) = 2\bigg( \dfrac{\tau_0}{1+c} \bigg)^2.
\end{equation}
Figure~\ref{fig:gauss1} represents the permutation importance measure of $X_1$ and its empirical counterpart versus the correlation $c$. The correlation $\tau_0$ is set to 0.7 and $c$ is varying between 0 and 1. We observe that the empirical permutation importance measure averaged over 100 simulations shares the same behaviour with the permutation importance measure (solid line in Fig.~\ref{fig:gauss1}) under predictor correlations. This is also highlighted when $\tau_0$ is varying for a fixed value $c$ (see Fig.~\ref{fig:gauss3}). In this case, the correlation $c$ is set to 0.5 and $\tau_0$ is varying between 0.1 and 0.85.

Secondly, we consider $p$ correlated predictors (Case 3) with $\tau_0=0.7$ and $c=0.5$. The permutation importance is given by Proposition \ref{lem:invC}:
\begin{equation} \label{eq:IXp}
I(X_j) = 2 \left( \dfrac{\tau_0}{1-c+pc} \right)^2.
\end{equation}
In Figure~\ref{fig:gauss2}, the permutation importance measure and its empirical version are drawn versus the number of correlated predictors chosen among a grid between 2 and 62. We observe again that the empirical permutation importance measure fits with the permutation importance measure (solid line in Fig.~\ref{fig:gauss2}). 

As mentioned earlier in the paper, the consistency of the empirical importance measure has not been    established yet for the case of correlated predictors. However, this  experiment suggests that  the results   given  in Section~\ref{sec:theo} are also valid for the empirical importance measure.

\subsection{Variable selection for classification and regression problems}

In this section, we compare the performances of the wrapper algorithms RFE and NRFE on five experiments.  More precisely, we illustrate in various situations that when the  predictors are highly correlated, NRFE algorithm provides  a screening of the predictors which is less efficient than the screening provided by RFE.

Experiences 1, 2 and 4 are borrowed from published contributions whereas Experiments 3 and 5 are original simulation problems. With these five experiments we embrace a variety of situations: small or high correlations between predictors, a few or many true predictors (predictors correlated with the outcome $Y$),  linear  links or non-linear links with the outcome. Except the last one, the experiments are carried out in the high dimensional setting. Over the five experiments, four are classification problems, since we were not able to give theoretical results about the permutation variable importance in this framework.

The model error corresponds to the misclassification rate for the classification and to the mean square error in regression case. As mentioned in Section~\ref{sec:algo}, the error is estimated in two different ways: the out-of-bag error embedded in the random forests and the error obtained using a validation set simulated independently. For each algorithm, we represent the validation errors and  the out-of-bag errors  in function of the number of variables in the model. In each case, we then select the minimum error model. Finally the error of the model selected  is evaluated on a test set simulated independently (see for instance  \cite{reunanen2003overfitting}). Note that when comparing the performances of both algorithms for fixed number a variables, it is not necessary to evaluate the error on an additional test sample.
The procedure is repeated 100 times to reduce the estimation variability. We also provide for each experiment the boxplots of the initial permutation importance measures, that is the importance measures used by NRFE for ranking the variables.

\subsubsection*{Experiment 1: small correlation between few predictors.}
The first problem involves only six true predictors in a classification problem inspired by \citet{rf:GPT10}. The aim of this experiment is to check that RFE and NRFE have similar performances in this simple context where there is only small correlation between a few true predictors.

\begin{figure}
	\begin{center}
	\includegraphics[scale= \echelle]{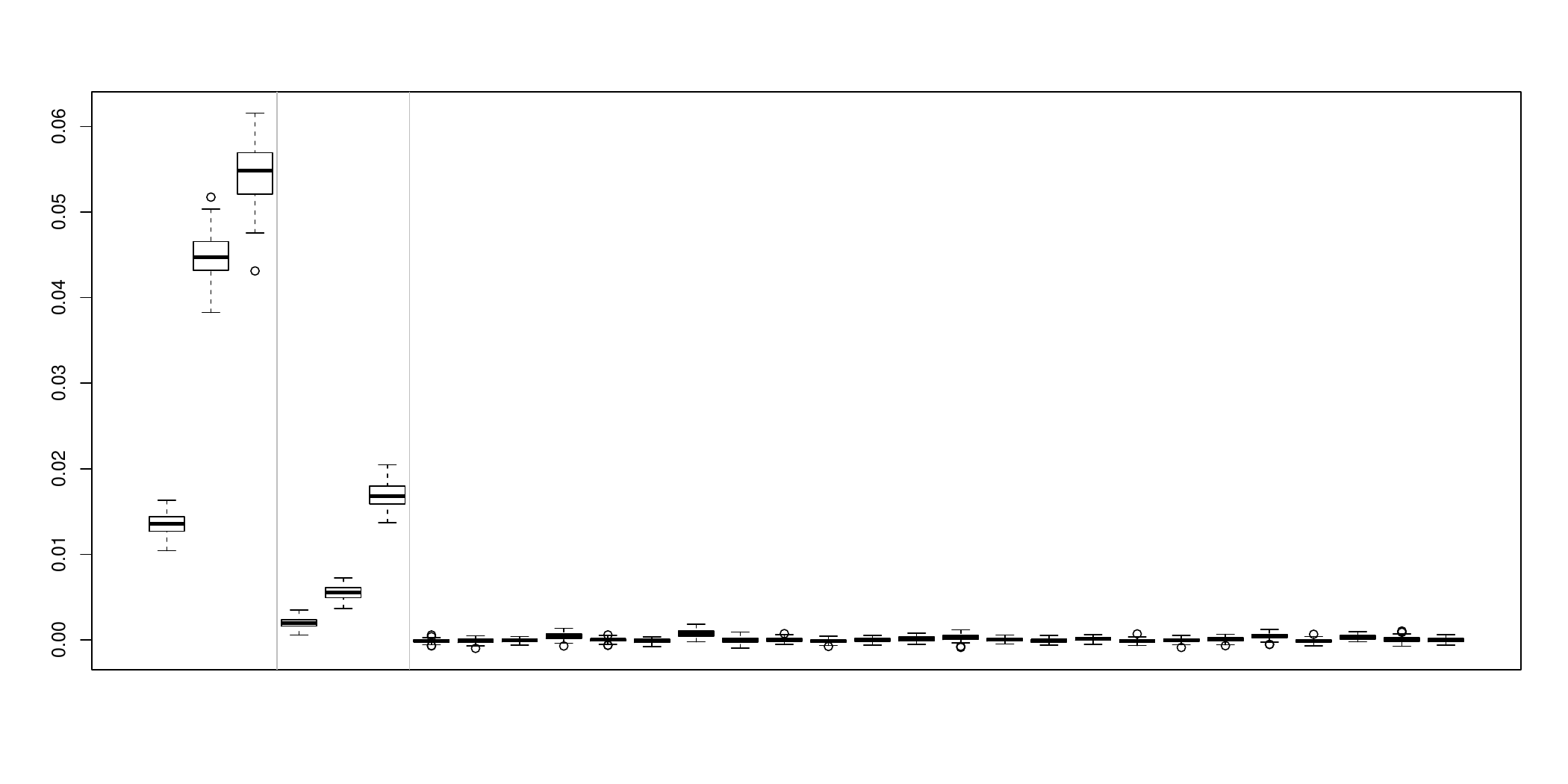}
	\caption{ Exp. 1 - Boxplots of the initial permutation importance measures. The boxplots of the  6 relevant variables together with the boxplots of the  24 first irrelevant variables are displayed.}
	\label{fig:weston_boxplots}
	\end{center}
\end{figure}

\noindent\textbf{Description.} The procedure generates two groups of three relevant variables respectively highly, moderately and weakly discriminant and a group of irrelevant variables. The relevant variables are drawn conditionally to a realisation of the outcome $Y$. More precisely, the first three relevant variables are generated from a Gaussian distribution $\mathcal{N}_{1}(Yj, 1)$ with probability 0.7 and from $\mathcal{N}_{1}(0,1)$ with probability 0.3, $j \in \{1, 2, 3\}$. The variables 4 to 6 are simulated from a distribution $\mathcal{N}_{1}(0,1)$ with probability 0.7 and from $\mathcal{N}_{1}(Y(j-3),1)$ with probability 0.3, $j \in \{4,5,6\}$ and are thus less predictive. The irrelevant variables are generated independently from the distribution $\mathcal{N}_{1}(0,20)$. Thus, conditionally to $Y$, the  $X_j$'s are drawn according to the following Gaussian mixtures densities:
$$p_{X_j}(x) = 0.7 \varphi(x; Yj, 1) + 0.3 \varphi(x; 0, 1), \; j \in \ensemble{1,2,3}, $$

$$p_{X_j}(x) = 0.7 \varphi(x; 0, 1) + 0.3 \varphi(x; Y(j-3), 1), \; j \in \ensemble{4,5,6},$$
and
$$p_{X_j}(x) = \varphi(x; 0, 20), \; j \in \ensemble{7, \ldots, p},$$
where $\varphi(\cdot; \mu, \sigma)$ is the normal density function with mean $\mu$ and standard error $\sigma$.
We generate $n=100$ samples and $p=200$ variables.

\noindent \textbf{Results.}  As shown by the boxplots of Figure~\ref{fig:weston_boxplots}, the permutation importance measures of the six relevant variables are much higher than for the other variables. The typical ranking proposed by these importances is the following : $X_3$, $X_2$, $X_6$, $X_1$, $X_5$, $X_4$  and then the other variables (see Fig. 4). This ranking corresponds to the simulation design and the behavior of the algorithms RFE and NRFE are quite similar in this context (see Figure~\ref{fig:weston}).  However, one can observe a slight difference between the two algorithms for the smallest models.  For instance, with two variables, the averaged validation error is  0.23 for NRFE but 0.042 for RFE. The presence of small correlations between the six first variables impacts the importance measures and leads NRFE to select, for instance, the variable $X_1$ before the variable $X_6$. Algorithm RFE is less subject to this phenomenon since it recomputes the permutation importance measures at each step.

\begin{figure}
	\begin{center}
	\subfloat[]{\includegraphics[scale=0.45]{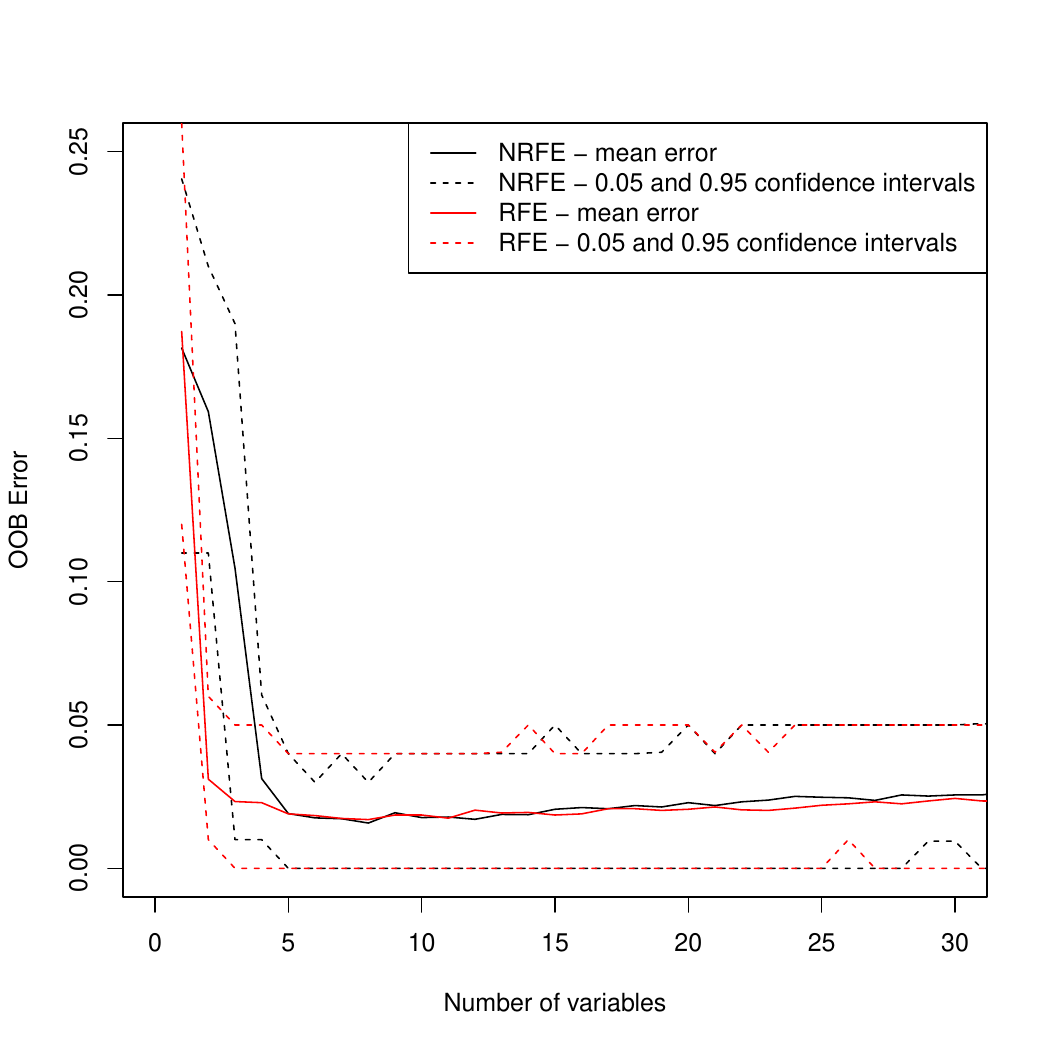}}
	\subfloat[]{\includegraphics[scale= 0.45]{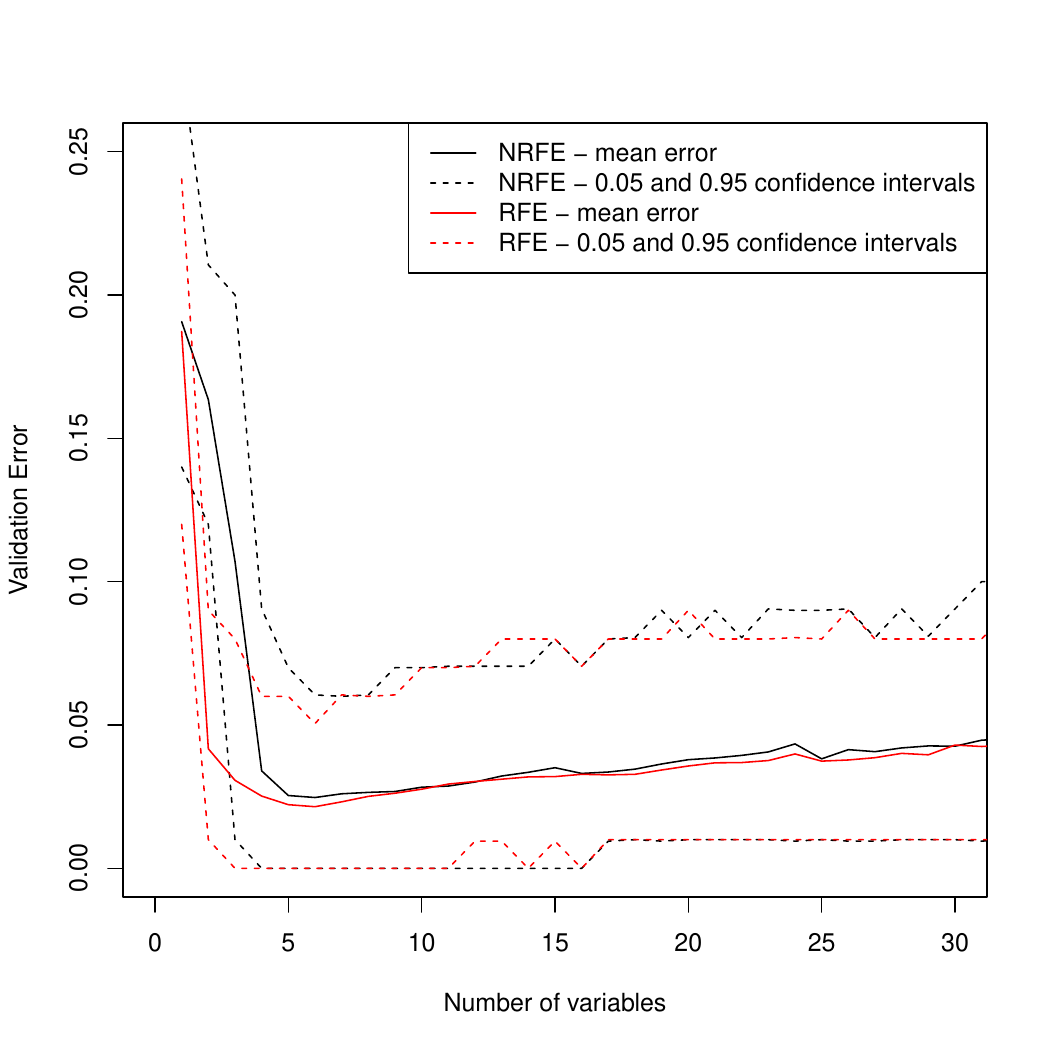}}
	\caption{ Exp. 1 - Out-of-bag error estimate (left) and validation set estimate (right) versus the number of variables for RFE and NRFE algorithms. The curves are averaged over 100 runs of variable selections.}
	\label{fig:weston}
	\end{center}
\end{figure}

\subsubsection*{Experiment 2: three large blocks of highly correlated variables}
We now compare RFE and NRFE on an experiment inspired by \citet{rf:TL11} and motivated by a genomic application. The aim of this experiment is to study the effect of correlation on the importance measure when large groups of true predictors are observed. Three groups of highly correlated predictors are involved in a classification setting. The link between the outcome and the predictors is quite simple and can be seen as slight perturbation of a linear discriminant protocol.

\begin{figure}
	\begin{center}
	\includegraphics[scale=\echelle]{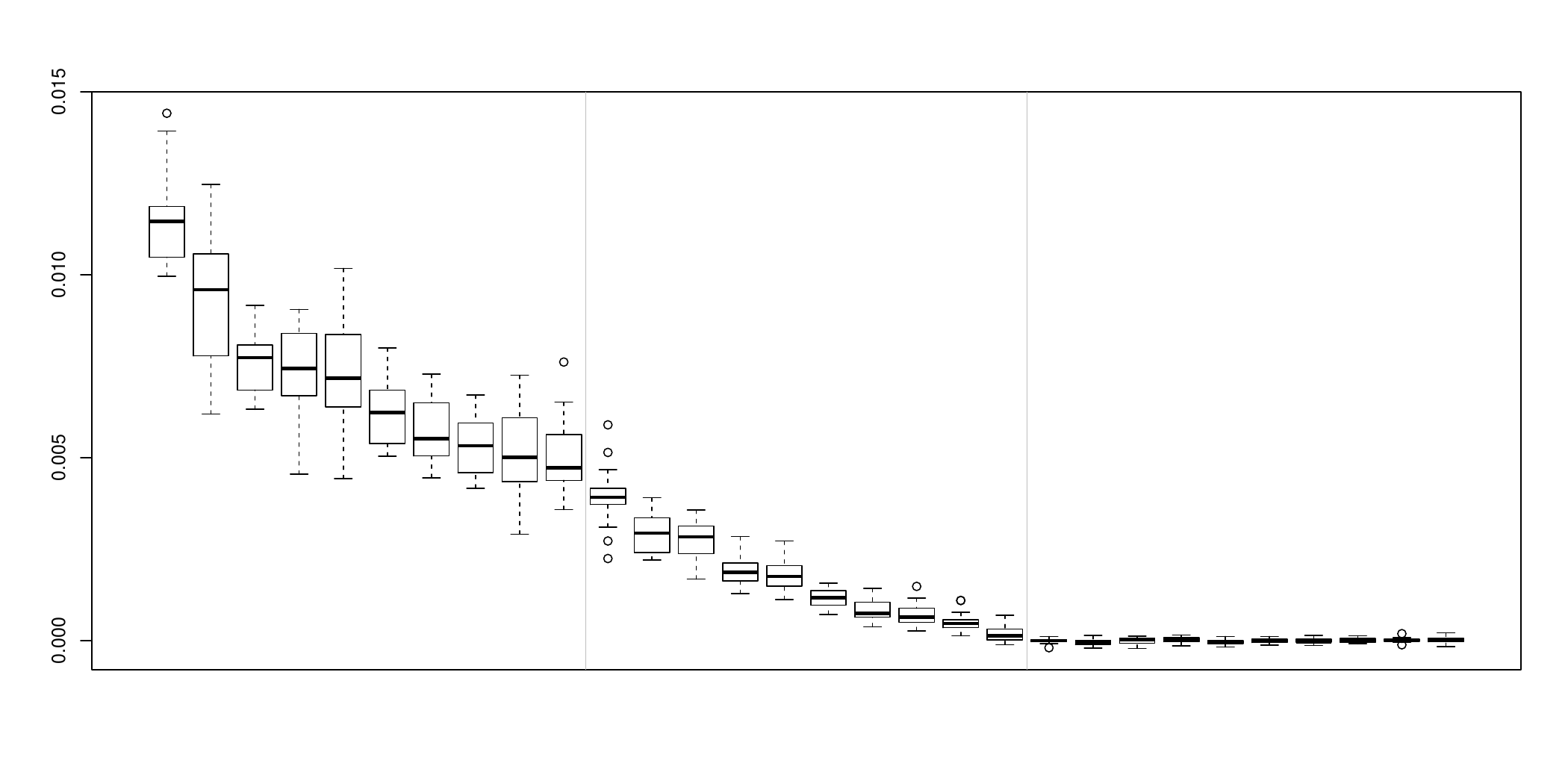}
	\caption{ Exp. 2 - Boxplots of the initial permutation importance measures. In each group, only the ten variables with the highest importances are displayed.}
	\label{fig:tolosi_boxplots}
	\end{center}
\end{figure}

\noindent\textbf{Description.}
Three groups which respectively contain $p_1$, $p_2$ and $p_3$  correlated variables are generated with a decreasing discriminative power: the variables are highly relevant in first group, they are weakly relevant in the second and irrelevant in the last group. Let $\mathbf U = (U_1, \dots,U_n)^t$ be a vector of i.i.d. variables drawn according to the mixture density $\frac{1}{2} \varphi(\cdot \, ; 0, 0.2) + \frac{1}{2}\varphi(\cdot \, ; 1, 0.3)$. For $j \in \{1, \dots ,p_1 \}$, let $\mathbf U^j$ a random vector defined by adding a Gaussian noise $\mathcal{N}_{1}(0,0.5)$ to 20~\% of the elements of $\mathbf U$, the perturbated coordinates being chosen at random. Note that by generating the $\mathbf U^j$'s this way, those are highly correlated. Independently, some vectors $\mathbf V,\mathbf V^1, \dots, \mathbf V^{p_2}$ and $\mathbf R$, $\mathbf R^1, \dots, \mathbf R^{p_3} $ are drawn in the same way. Finally, the  outcomes $Y_i$'s are defined for $i \in \{1, \dots,n \}$ by \begin{equation}
Y_i = \left\{
  \begin{array}{rl}
    1 & \mbox{ if } 5U_i + 4V_i - (\overline{5\mathbf U + 4\mathbf V}) + \varepsilon_i > 0 \\
    0 & \mbox{ otherwise} \\
  \end{array}
\right. ,
\end{equation}
where $\overline{5\mathbf U + 4\mathbf V}$ denotes the mean of the vector $5\mathbf U + 4\mathbf V$ and the $\varepsilon_i$'s are i.i.d. random variables drawn according to $\mathcal{N}_{1}(0,0.1)$. The problem considered here consists in predicting $Y$ using as predictors all the $\mathbf U^j$, the $\mathbf V^j$ and the $\mathbf R^j$, but not $\mathbf U$, $\mathbf V$ and $\mathbf R$. Since only $\mathbf U$ and $\mathbf V$ are involved in the definition of $Y$, it follows that only the $\mathbf U^j$ and the $\mathbf V^j$ are relevant. Moreover, the $\mathbf U^j$  are more relevant than the $\mathbf V^j$. We set the number of observations to $n=100$, the number of variables to $p=250$ with $p_1= p_2 = 100$ and $p_3 = 50$.

\noindent\textbf{Results.}
Figure \ref{fig:tolosi} shows a meaningful difference between RFE and NRFE regarding the number of involved variables. The test errors for NRFE  and  RFE  are similar (see Table~\ref{table:err}) but, as shown by the validation error curves and the out-of-bag errors curves on Figure~\ref{fig:tolosi}, NRFE  selects around 60 variables whereas RFE procedure selects less than 20 variables. By recomputing the importances at each step of the algorithm,  RFE  provides a screening of variables which is much more efficient than the one obtained by NRFE.
\begin{figure}
	\begin{center}
	\subfloat[]{\includegraphics[scale=0.45]{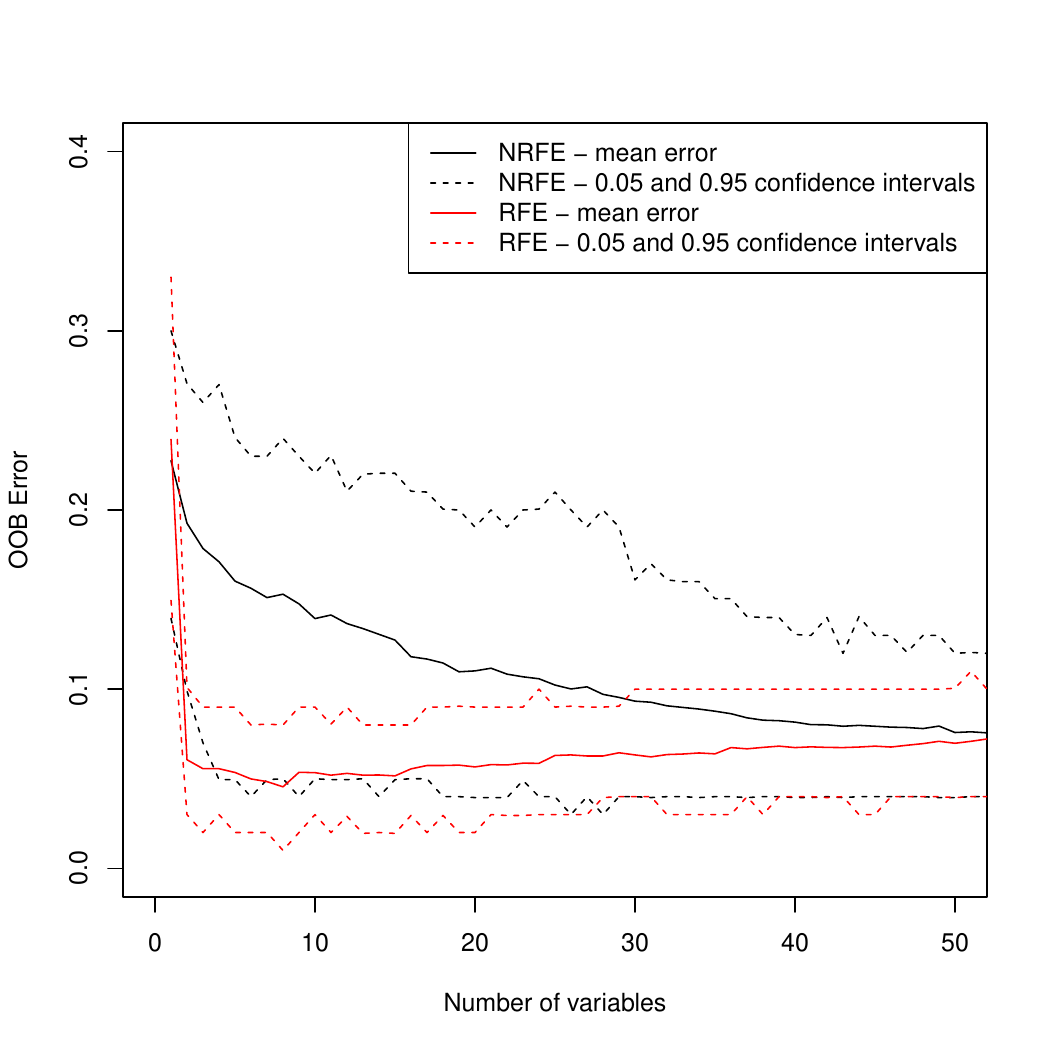}}
	\subfloat[]{\includegraphics[scale=0.45]{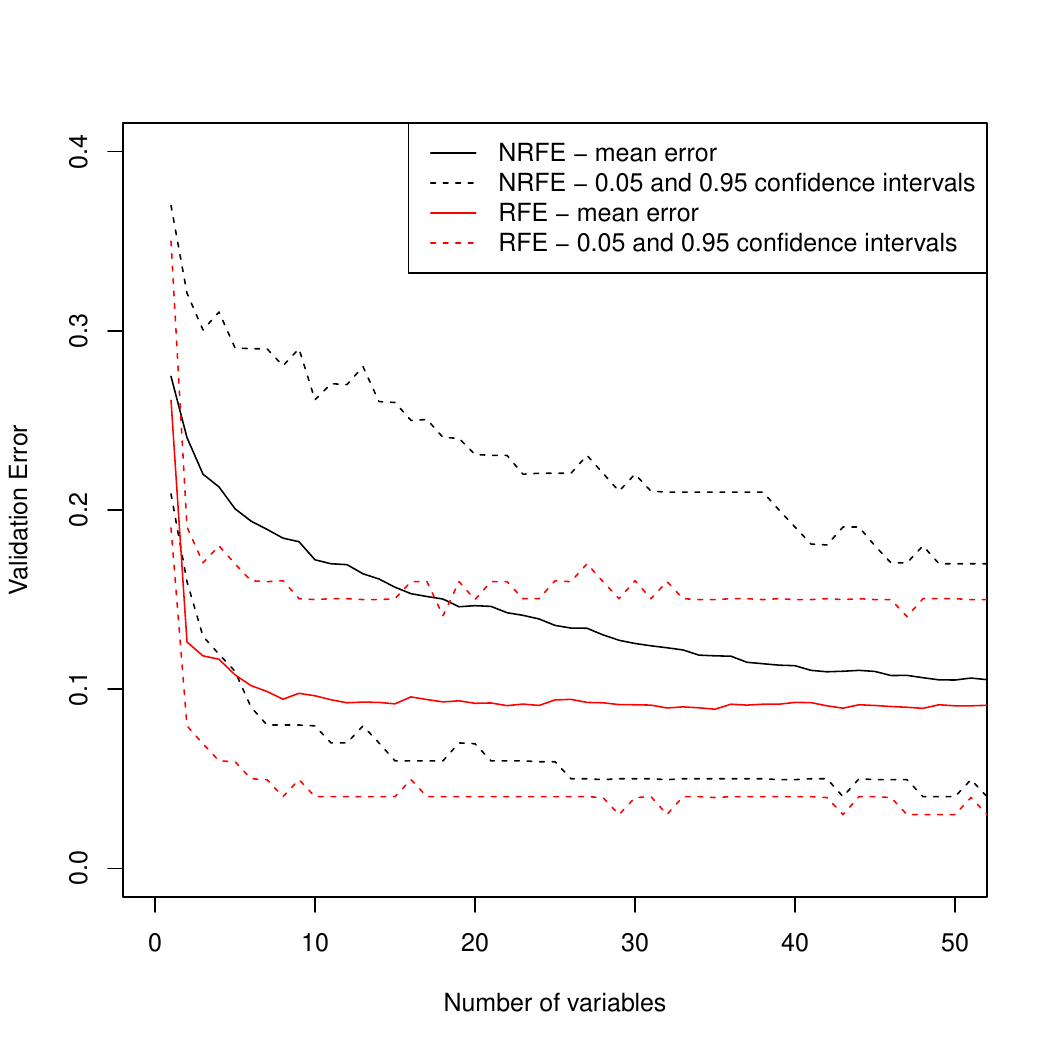}}
	\caption{ Exp. 2 - Out-of-bag error estimate  (left) and validation set estimate  (right) versus the number of variables for RFE and NRFE algorithms. The curves are averaged over 100 runs of variable selections.}
	\label{fig:tolosi}
	\end{center}
\end{figure}

\subsubsection*{Experiment 3: four blocks of highly correlated predictors and one block of independent predictors}

\begin{figure}
	\begin{center}
	\includegraphics[scale=\echelle]{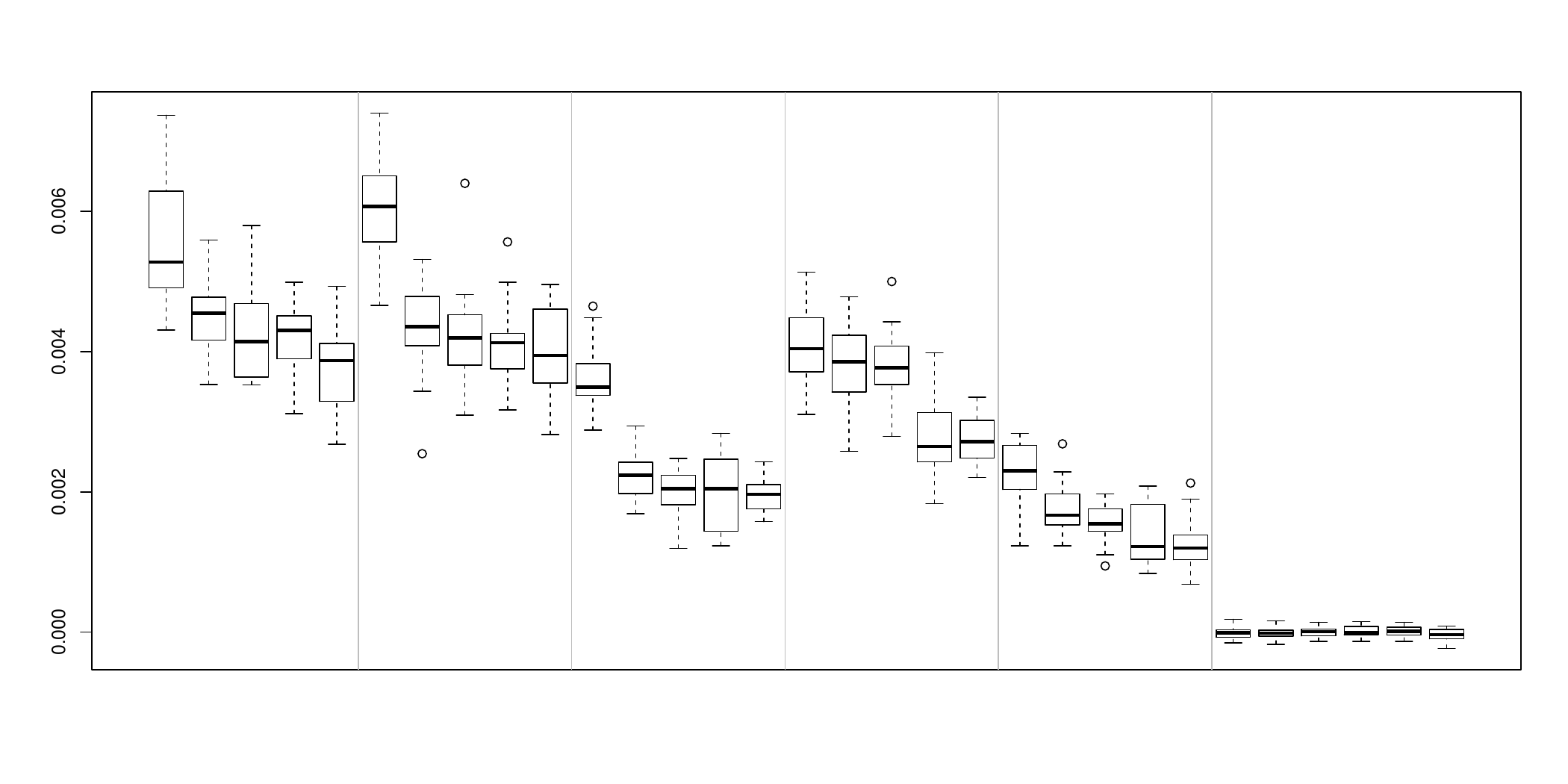}
	\caption{ Exp. 3 - Boxplots of the initial permutation importance measures. In each group, only the five variables with the highest importances are displayed.}
	\label{fig:mixture_boxplots}
	\end{center}
\end{figure}

This experiment is an original classification problem based on Gaussian mixture distributions with a large number of predictors. The main difference with the experiment 2 is that four groups of highly correlated predictors and a group of independent predictors are now involved. The aim is to illustrate the situation where some independent variables are less predictive than other variables but have higher importances due to the correlation effect. These independent variables can then be selected by NRFE whereas more predictive variables are not kept in the selection.

\noindent\textbf{Description.} The outcome $Y$ is sampled from  a Bernoulli distribution of parameter $1/2$. We generate $n_b$ groups $\Bcal_1, \ldots, \Bcal_{n_b}$ of correlated variables and an additional group $\Bcal_{ind}$ of independent variables. Within a group $\Bcal_\ell$, the vector of variables is simulated from  a Gaussian mixture.  More precisely,  conditionally to $Y$  the vector of variables in  $\Bcal_\ell$ has the distribution of a multivariate Gaussian distribution $\mathcal{N}_q(Y \mu_\ell \indd ,C)$ with $C = (1-c)Id + c \indd \indd^t$ where $c$ is the correlation between two different variables of $\Bcal_\ell$. Moreover, conditionally to $Y$ a random variables $X_j$ of $\Bcal_{ind}$ has a Gaussian distribution $\mathcal{N}(Y \mu_j  ,1)$. The parameters $\mu_\ell$ and $\mu_j$ correspond to the discriminative power of the variables. We choose these mean parameters decreasing linearly from 1 to 0.5. In this way, the groups $\Bcal_1, \ldots, \Bcal_{n_b}$ have a decrease discriminative power which are higher than the independent group $\Bcal_{ind}$.  For the experiment, we simulate $n=250$ samples and $p=500$ variables: 4 blocks of $q=15$ variables highly correlated with $c=0.9$, one block $\Bcal_{ind}$ of 10 independent variables and 430 irrelevant variables.

\noindent\textbf{Results.}
\begin{figure}
	\begin{center}
	\subfloat[]{\includegraphics[scale=0.45]{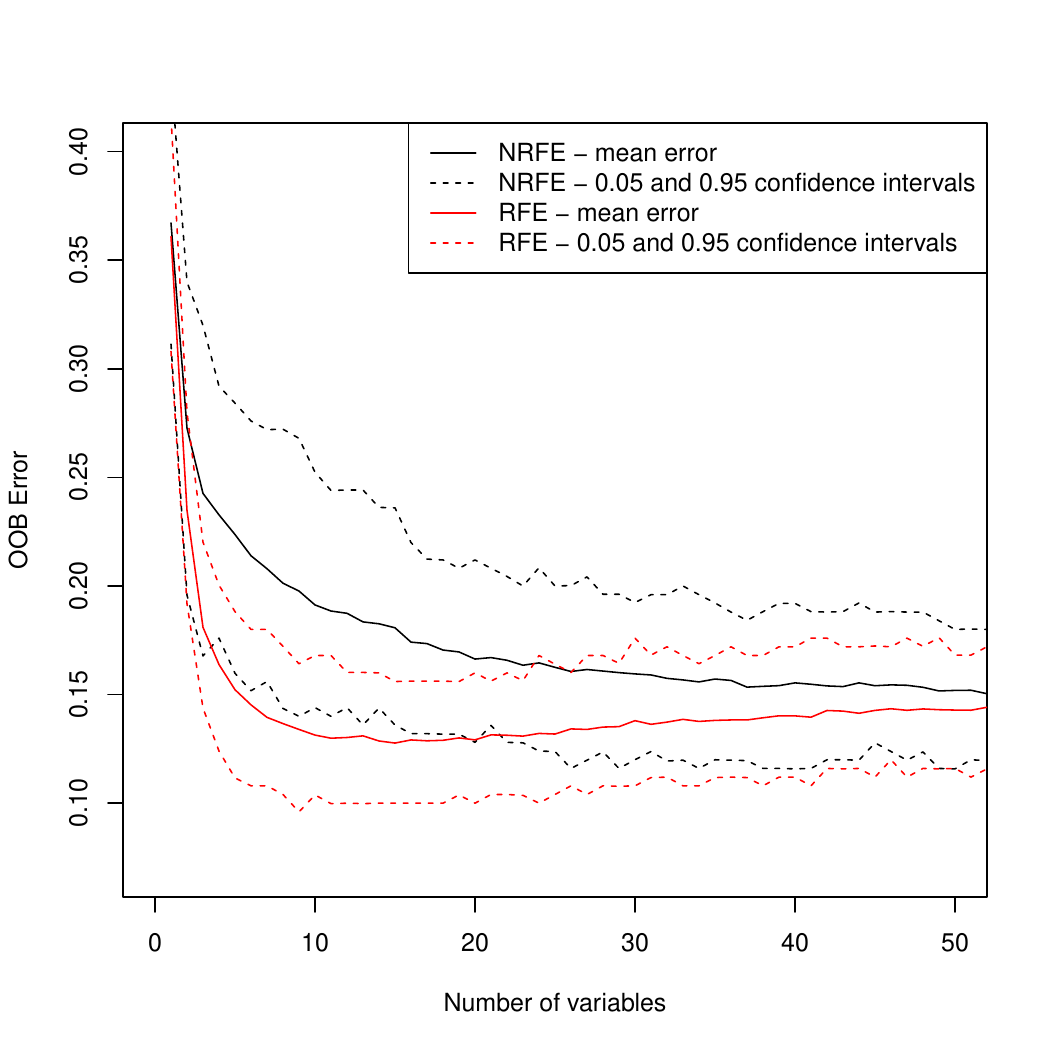}}
	\subfloat[]{\includegraphics[scale=0.45]{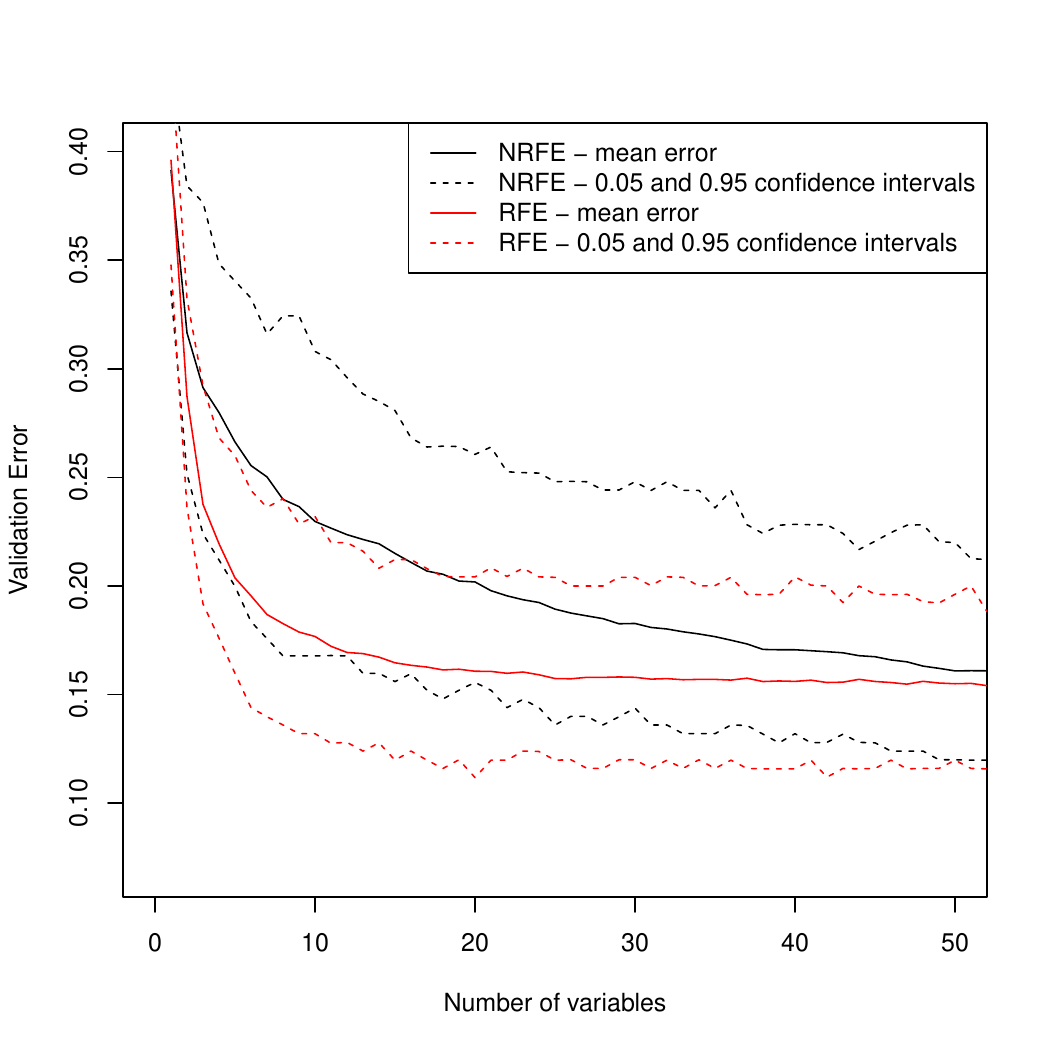}}
	\caption{ Exp. 3 - Out-of-bag error estimate  (left) and validation set estimate  (right) versus the number of variables for RFE and NRFE algorithms. The curves are averaged over 100 runs of variable selections.}
	\label{fig:mixture}
	\end{center}
\end{figure}
As for experiment 2, RFE provides an efficient screening of the variables: the validation error and the out-of-bags error decrease much faster for RFE than for NRFE on Figure~\ref{fig:mixture}. The validation error and the out-of-bag error for the models proposed by RFE with 20 variables are comparable to the test error of the minimum error model selected by RFE (see Table~\ref{table:err}).  This is not the case for NRFE. Again, this difference is explained by the impact of the high correlations between predictors on the ranking of the importances. For instance, one can observed that an importance in group 5 (group of independent variables) is higher than the importances of group 3 (see Fig. \ref{fig:mixture_boxplots}) whereas the predictive abilities of the independent group are simulated to be lower. The NRFE procedure tends to select variables from the independent group before selecting variables from group 3.

\subsubsection*{Experiment 4: many large blocks of correlated variables}

This classification problem inspired by \citet{rf:AK08} is related to gene expressions. The aim is to compare the algorithms in a difficult situation involving many groups composed of many correlated predictors. The correlation between predictors  is the same in a group but it  increases from one group to another. Moreover the link between the predictors and the outcome is more complex than in the previous experiments since only one variable of each group is used to define the outcome.

\begin{figure}
	\begin{center}
	\includegraphics[scale=\echelle]{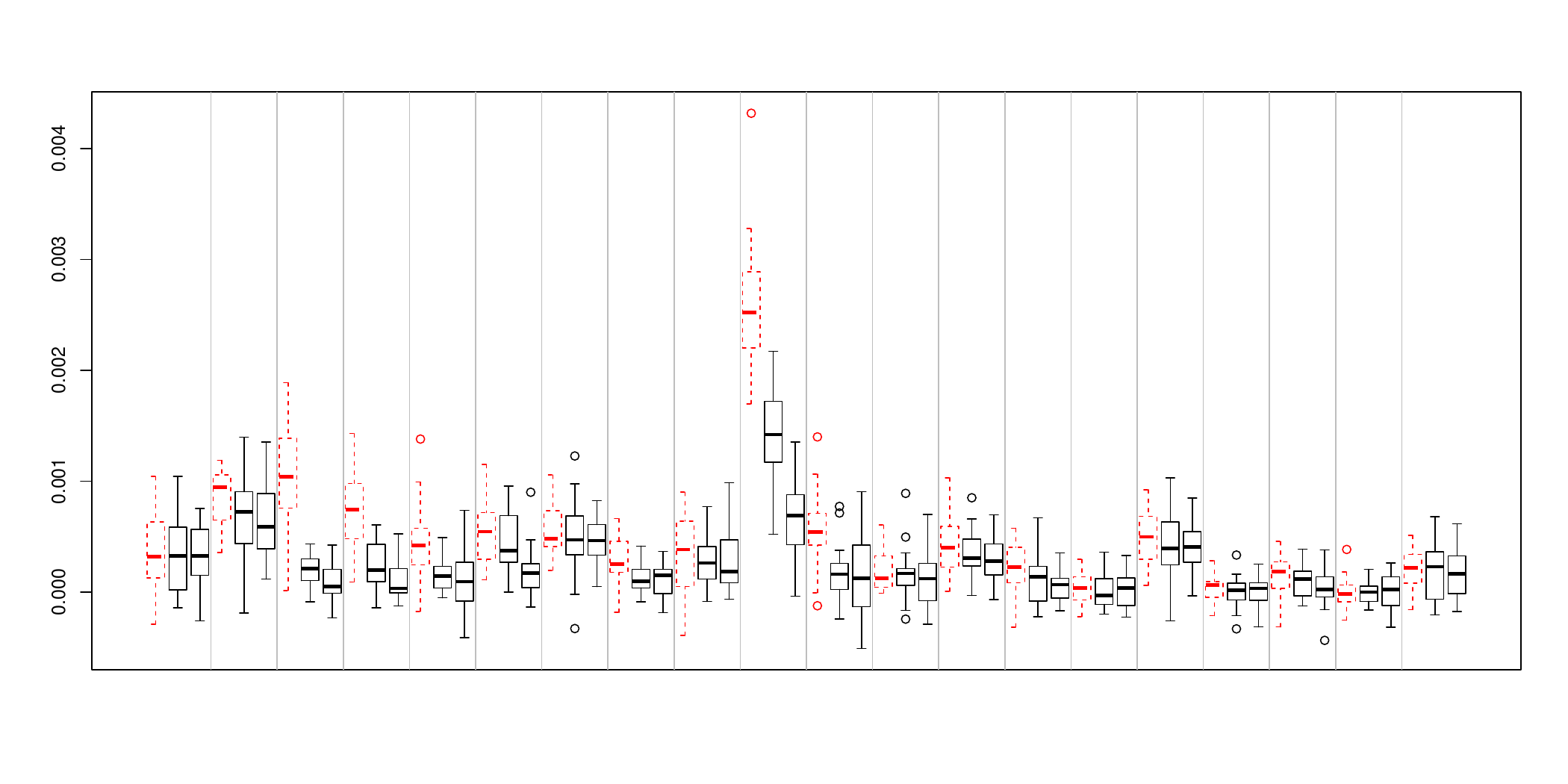}
	\caption{ Exp. 4 - Boxplots of the initial permutation importance measures. For each group, only the predictive variable (dashed lines) and the two variables with the highest importance in the same group (solid lines) are displayed.}
	\label{fig:archer_boxplots}
	\end{center}
\end{figure}

\noindent\textbf{Description.} We simulate $n=100$ independent vectors $\mathbf X_1, \ldots, \mathbf X_n$ as follows: the mean vectors $\mu_i$'s are chosen uniformly distributed in $[6; 12] ^p $ and the $\mathbf X_i$'s are drawn from the Gaussian distribution $\mathcal N _p(\mu_i,C)$.  Each observation is composed of $L=20$ independent blocks of $K=40$ covariates ($p=800$). The corresponding covariance matrix of the observed complete vector has the block-diagonal form
$$C = \begin{pmatrix}
		  	C_1 	& 0 		& \cdots 	& 0 \\
			0 		& C_2 		& \cdots	& 0 \\
			\vdots 	& \vdots	& \ddots	& \vdots \\
			0 		& 0 		& \cdots	& C_L
\end{pmatrix}.$$
The covariance matrix $C_\ell$ of the $\ell$-th group is
$$C_\ell = \begin{pmatrix}
		  	1 		& \rho_\ell 		& \cdots 	& \rho_\ell \\
			\rho_\ell 	& 1 		& \cdots	& \rho_\ell \\
			\vdots 	& \vdots	& \ddots	& \vdots \\
			\rho_\ell 	& \rho_\ell 		& \cdots	& 1
\end{pmatrix},$$
where the correlations are set to $\rho_\ell = 0.05\ell - 0.05$. The correlation of each block is taken from 0 to 0.95 by increments of 0.05. We then compute the probability that the observation $\mathbf X_i$ is from class 1 by
$$\pi_i = \dfrac{e^{\mathbf X_i^t \beta}}{1+e^{\mathbf X_i^t \beta}},$$
for $i \in \{1, \dots,n \}$. The vector $\beta$ of regression coefficients is sparse: $\beta_j = 0.5$ if $j = (\ell-1)K+1$ for some $\ell \in \{1, \dots,L\}$ and 0 otherwise. In other words, the posterior probability $\pi_i$ is generated using only informations from the first variables of each group. Finally, for $i \in \{1, \dots,n \}$ the response $Y_i$ is generated from
$$Y_i = \left\{
  \begin{array}{rl}
    1 & \mbox{ if } \pi_i < U_i, \: \textrm{ where } U_i \textrm{ is an uniform distribution on } [0,1] \\
    0 & \mbox{ otherwise}. \\
  \end{array}
\right.
$$

\noindent\textbf{Results.}
\begin{figure}
	\begin{center}
	\subfloat[]{\includegraphics[scale=0.45]{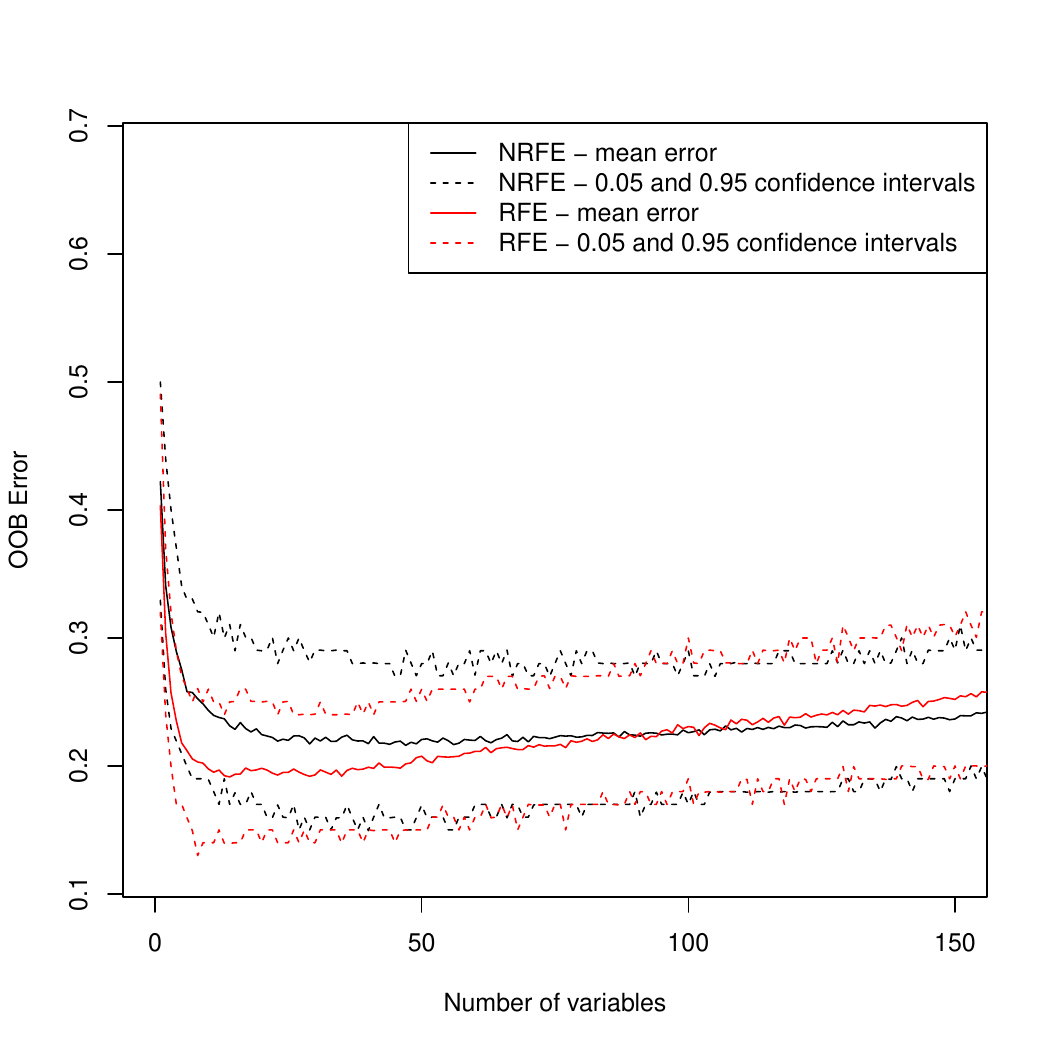}}
	\subfloat[]{\includegraphics[scale=0.45]{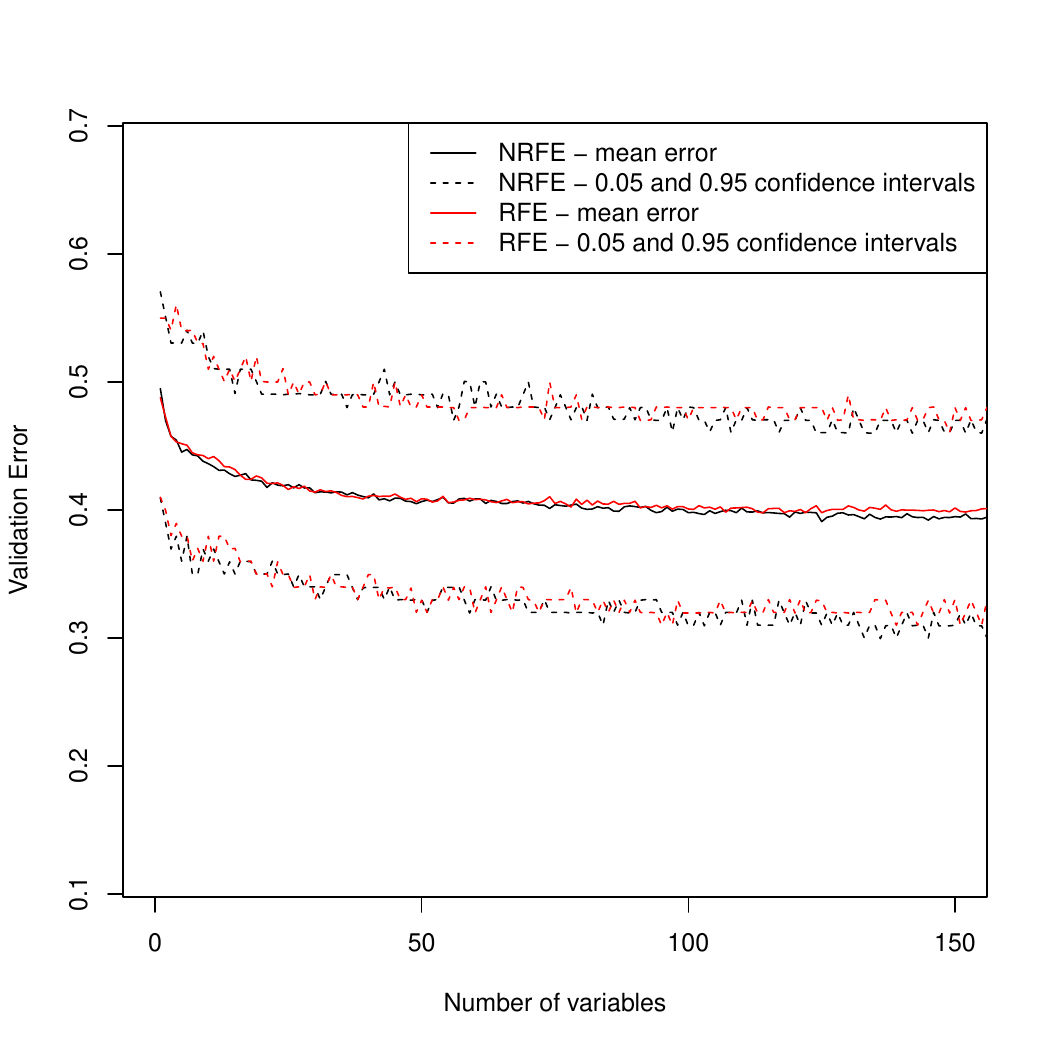}}
	\caption{ Exp. 4 - Out-of-bag error estimate  (left) and validation set estimate  (right) versus the number of variables for RFE and NRFE algorithms. The curves are averaged over 100 runs of variable selections.}
	\label{fig:archer}
	\end{center}
\end{figure}
This is a tricky problem in a high dimensional setting (100 observations of 800 variables). Without surprise, it is a difficult task to find the relevant variables (in dashed line) from the empirical permutation importances (see Fig.~\ref{fig:archer_boxplots}). In particular, we cannot clearly discriminate the relevant variables of each group as they do not necessary show higher importances than the other variables in the same block. In this complex simulation design, the screening of variables provided by RFE and NRFE are similar according to the validation error whereas RFE seems more efficient according to OOB error (see Fig.~\ref{fig:archer}).

\subsubsection*{Experiment 5. Under the assumptions of Proposition \ref{prop:Gauss} in high dimension.}
\begin{figure}
	\begin{center}
	\includegraphics[scale= 0.5 ]{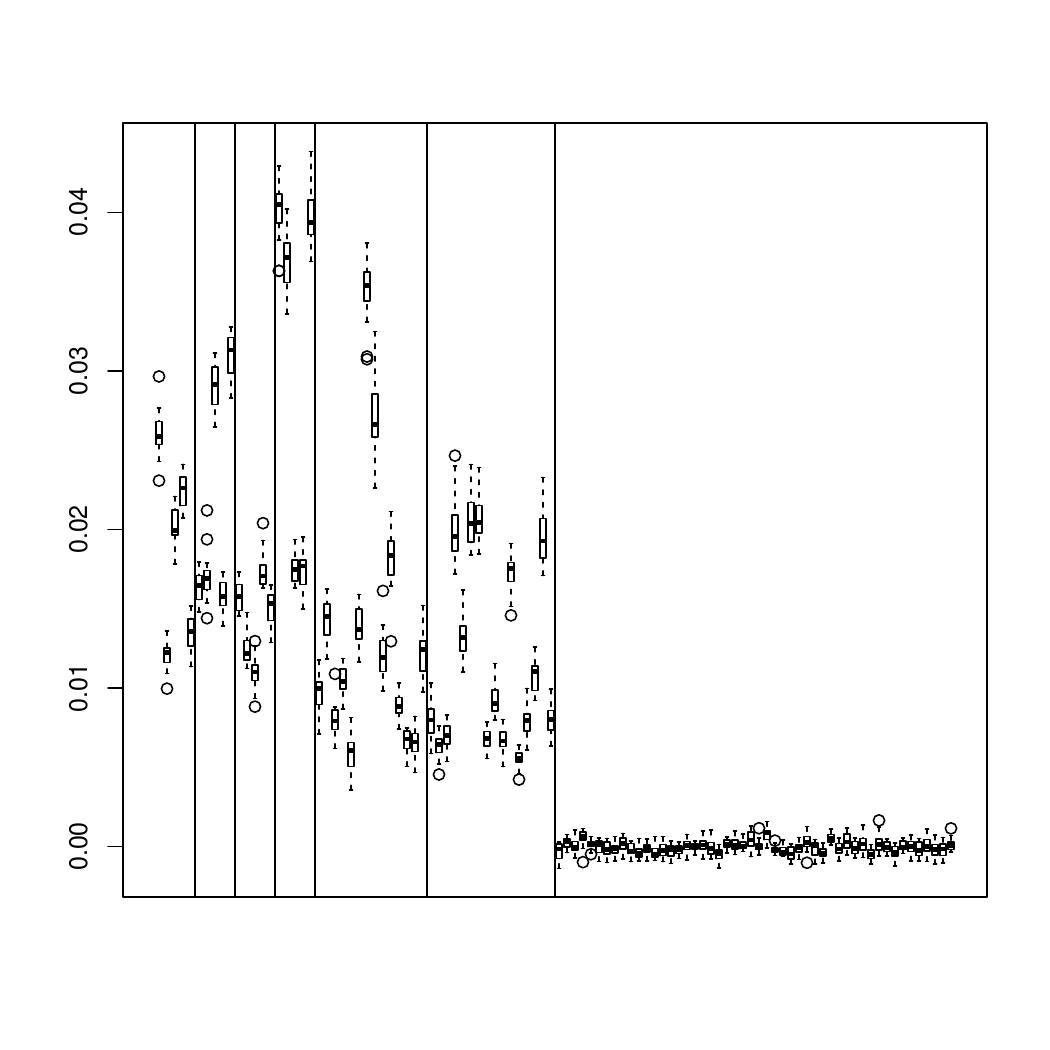}
	\caption{ Exp. 5 - Boxplots of the initial permutation importance measures. Four groups of 5 correlated variables and two groups of 15 correlated variables are simulated.}
	\label{fig:reg_boxplots}
	\end{center}
\end{figure}

We now propose a complex linear regression design involving Gaussian variables which satisfy the assumptions of Proposition~\ref{prop:Gauss}. Many groups of correlated variables are considered. The sizes of the groups are different but all the variables have the same predictive power. The correlation between the variables in the groups is the same for all the groups. The goal of this experiment is to observe the effect of a large number of correlated predictors on the performance of the two algorithms.

\noindent\textbf{Description.}
We simulate $n=1000$ i.i.d. copies of the random vector $(\mathbf X, Y)$ from the multivariate Gaussian distribution $\mathcal N _{p+1}\bigg(0, \begin{pmatrix}
	C & \boldsymbol \tau\\
	\boldsymbol \tau^t & 1
\end{pmatrix} \bigg)$
where the vector $\boldsymbol \tau = (\tau_0, \ldots, \tau_0, 0, \ldots, 0)^t$ contains the covariances between each predictor $X_j$ and $Y$, $\tau_j=\tau_0$ for the relevant variables and $\tau_j=0$ for the irrelevant ones.
We consider $L$ groups of correlated variables and some additional irrelevant and independent variables.
The matrix $C$ has a block-diagonal form
$$C = \begin{pmatrix}
		  	C_1 	& 0 		& \cdots 	& 0 		& 0\\
			0 		& C_2 		& \cdots	& 0 		& 0\\
			\vdots 	& \vdots	& \ddots	& \vdots 	& \vdots\\
			0 		& 0 		& \cdots	& C_L		& 0\\
			0 		& 0 		& \cdots	& 0			& Id\\
\end{pmatrix},$$
where
$$C_\ell = \begin{pmatrix}
		  	1 		& \rho 		& \cdots 	& \rho \\
			\rho 	& 1 		& \cdots	& \rho \\
			\vdots 	& \vdots	& \ddots	& \vdots \\
			\rho 	& \rho 		& \cdots	& 1
\end{pmatrix}.$$
The blocks are simulated with different size in order to highlight the effect of the number of correlated variables. We take 4 groups of 5 variables, 2 groups of 15 variables and 50 irrelevant variables. The correlation $\rho$ is set to 0.9 and $\tau_0$ is equal to 0.3 for the relevant variables and is equal to 0 for the irrelevant ones.

\noindent\textbf{Results.}
For the same reasons as for the previous experiments (see Fig.~\ref{fig:reg_boxplots}),  NRFE mainly selects the variables from the four small groups at first  (5 correlated variables) and it does not select enough variables from the two large groups (15 correlated variables).  On the contrary, RFE selects earlier the two large groups. This problem illustrates the difficulties of NRFE to select the most predictive variables in a regression  setting.
\begin{figure}
	\begin{center}
	\subfloat[]{\includegraphics[scale= 0.45 ]{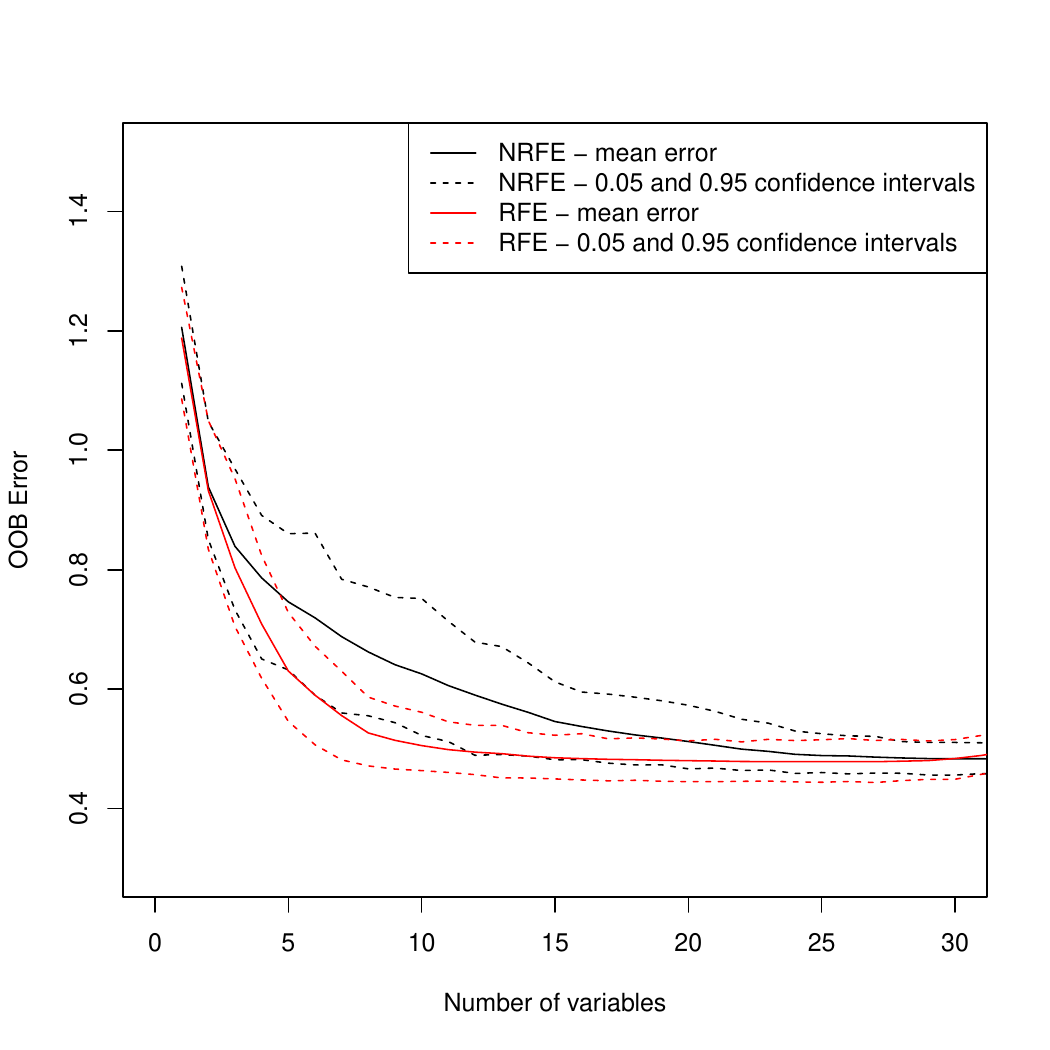}}
	\subfloat[]{\includegraphics[scale=0.45]{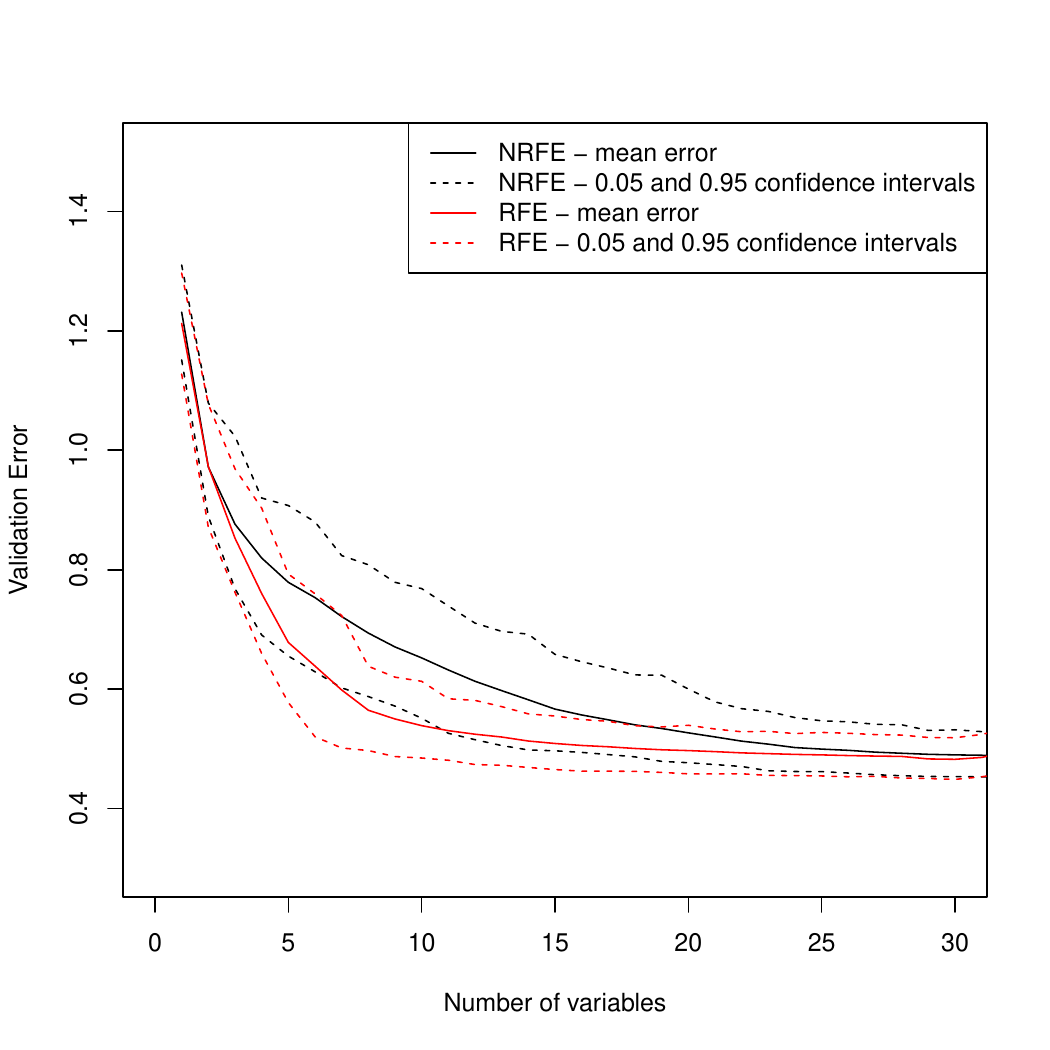}}
	\caption{ Exp. 5 - Out-of-bag MSE error (left) and validation set estimate (right) versus the number of variables for RFE and NRFE algorithms. The curves are averaged over 100 runs of variable selections.}
	\label{fig:reg}
	\end{center}
\end{figure}

\subsubsection*{Prediction errors of the five experiments}

Table~\ref{table:err} summarizes the performances of the two variable selection approaches on a sample test. The errors are given for a model which minimizes the  error among the selection path induced by the backward search (Minimum error model). We also consider the errors of the small models defined  with only  five predictors. We give the results of  two possible methods for estimating the errors : the train+validation+set method and the out-of-bag+test method. The errors are averaged over 100 runs. 

RFE and NRFE show similar performances  for the minimum error model. However, when considering a model with five variables, the minimum error obtained with RFE  is significantly smaller than the NRFE in experiments  2, 3  and 5. Indeed, as noticed before, RFE is particularly useful to reduce the effect of correlation in the last steps of the variable selection procedure and tends to keep the variables which are the most able to predict the outcome.
 
 These simulations also confirm the fact that the correlation increases the instability of the empirical permutation importance, as explained in \citet{rf:TL11} and \citet{rf:GPT10}. 
\begin{table}
\begin{center}
\begin{tabular}{cl|cc|ccc}
	& & \multicolumn{2}{c|}{Minimum error model (std)} & \multicolumn{2}{|c}{Model with five variables (std)}\\
	& & RFE & NRFE & RFE & NRFE\\
	\hline
	\hline
  	\multirow{5}{*}{oob-test} & Exp. 1 & 0.0316 (0.026) & 0.028 (0.020) & 0.0309 (0.026) & 0.0237 (0.018)\\
  					& Exp. 2 &  0.111 (0.038) & 0.104 (0.040) &  0.118 (0.038) & 0.188 (0.052)\\
  					& Exp. 3 & 0.166 (0.026) & 0.165 (0.032) &  0.206 (0.032) & 0.270 (0.047)\\	
  					& Exp. 4 & 0.437 (0.052) & 0.414 (0.052) & 0.463 (0.054) & 0.463 (0.046)\\
  					& Exp. 5 & 0.493 (0.024) & 0.490 (0.025) & 0.681 (0.069) & 0.778 (0.078)\\					
  	\hline
	\multirow{5}{*}{train-valid-test} 	& Exp. 1 & 0.0285 (0.022) & 0.025 (0.021) & 0.0284 (0.022) & 0.0245 (0.019)\\
  					& Exp. 2 & 0.0977 (0.038) &0.0981 (0.035) & 0.115 (0.038) & 0.187 (0.053)\\
  					& Exp. 3 & 0.154 (0.026) & 0.159 (0.028) & 0.207 (0.030) & 0.270 (0.047)\\	
  					& Exp. 4 & 0.418 (0.055) & 0.410 (0.057) & 0.462 (0.052) & 0.461 (0.045)\\
					& Exp. 5 & 0.489 (0.023) & 0.489 (0.025) & 0.681 (0.069) & 0.778 (0.078)\\
  	\hline
\end{tabular}
\end{center}
\caption{Averaged test errors (with standard deviations) over 100 runs for the minimum error models and for the models with only five variables.}
\label{table:err}
\end{table}

\section{Application to the Landsat Satellite dataset}

In this section,  the RFE and the NRFE algorithms  are compared on the Landsat Satellite dataset. The data consists of the multi-spectral values of pixels in a sub-area of a satellite image. It is available at the UCI Machine Learning Repository (see \url{http://www.ics.uci.edu/~mlearn/}). 

The training set contains $n=4435$ samples. Each observation is a vector of size $p=36$, composed by the pixel values in four spectral band (two in the visible region and two in the near infra-red) of each of the 9 pixels in the $3\times3$ neighbourhood.
The first four variables are the spectral values for the top-left pixel, the four next are the spectral values for the top-middle pixel and then those for the top-right pixel, and so on with the pixels read out in sequence left-to-right and top-to-bottom.
The aim is to predict the classification label associated with the central pixel in each neighbourhood, given the multi-spectral values. The labels are coded as follows: (1) red soil (2) cotton crop (3) grey soil (4) damp grey soil (5) soil with vegetation stubble (6) mixture class (all types present) (7) very damp grey soil.

Note that the four spectral values for the central pixels, which is the most correlated to the class labels, are given by variables 17, 18, 19 and 20. The variables corresponding to the same spectral band are 
particularly correlated.  Consequently, the variables in this  dataset can be grouped into four blocks of nine correlated variables $G_{17}$, $G_{18}$, $G_{19}$ and $G_{20}$, each block corresponding to one of the variables 17, 18, 19 and 20.

We run the two  algorithms on the data 100 times. The averaged OOB error and validation errors are given on Figure~\ref{fig:landsat}. For each run, the validation error has been computed  by randomly selecting  a validation set containing one third of the observations.

On this dataset, the screening of variables proposed  by RFE is more efficient than the screening proposed by NRFE.  The OOB error (see Fig.~\ref{fig:landsatdataOOB}) and the validation error (see Fig.~\ref{fig:landsatdataTEST}) are both lower for RFE  than for NRFE. The superiority of  RFE is particularly manifest for small size models: with five variables, the OOB error rate is 0.13 for RFE (with standard deviation 0.001) while it is 0.48 for NRFE (with standard deviation 0.003). The validation error is 0.13 (with standard deviation 0.008) for RFE while it is 0.48 (with standard deviation 0.01) for NRFE.

\begin{figure}
	\begin{center}
	\subfloat[]{\includegraphics[scale=0.45, page=1]{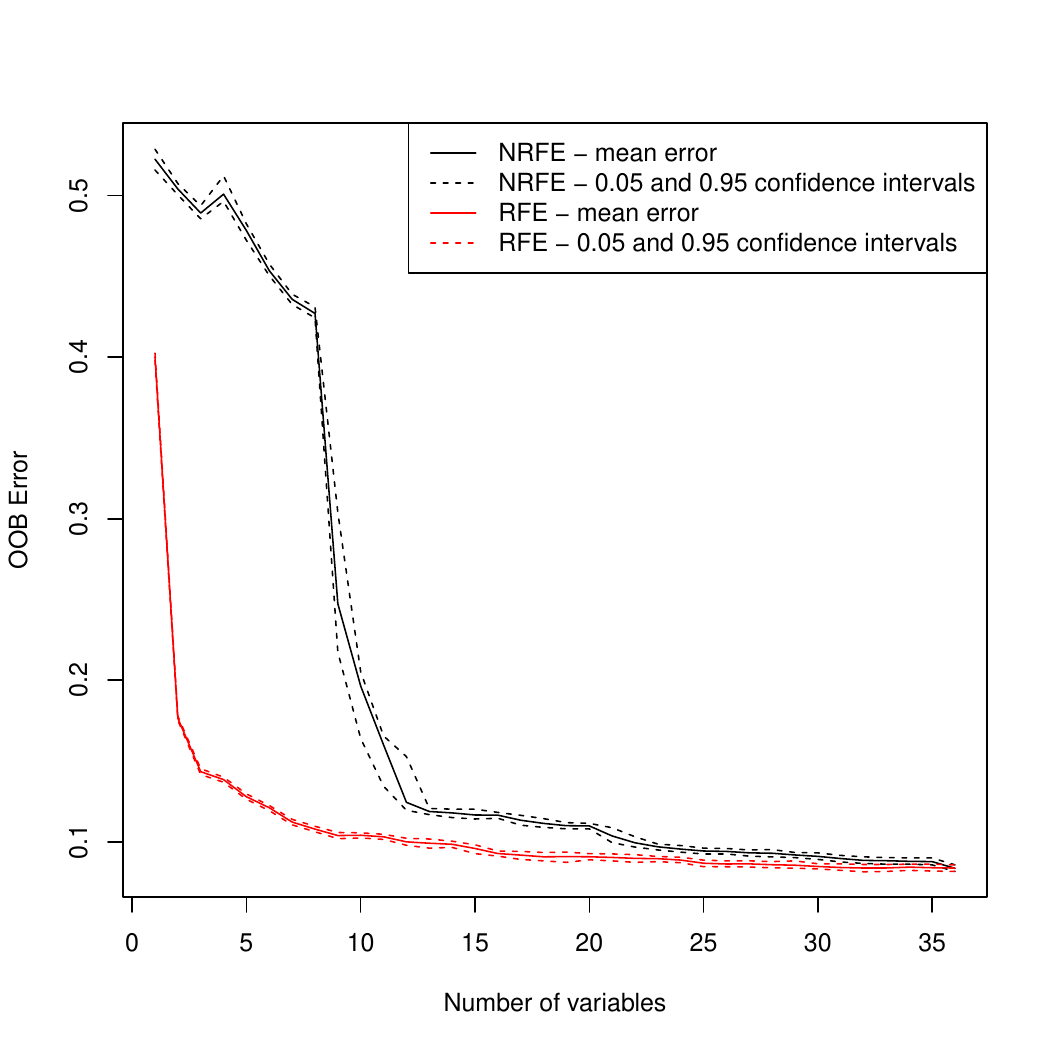}\label{fig:landsatdataOOB}}
	\subfloat[]{\includegraphics[scale=0.45, page=2]{satelliteErrors.pdf}\label{fig:landsatdataTEST}}
\caption{ Landstat data - Out-of-bag error estimate (left) and validation set estimate  (right) versus the number of variables for RFE and NRFE algorithms. The curves are averaged over 100 runs of variable selections.}
	\label{fig:landsat}
	\end{center}
\end{figure}

The proportion of times each variable has been selected in the set of the ten first selected variables  is given on  Figure~\ref{fig:barplotLandsat}. For both algorithms, the variables 17, 18 and 20 are always selected. The variable 19 is  selected 54 times over 100 by RFE whereas it is never selected by NRFE, which is natural according to the ranking of the importance measures  (see Fig.~\ref{fig:boxplotLandsat}). The situation is actually even worst: Figure~\ref{fig:Selec10Landsat} shows that, over the 100 runs of the NRFE algorithm, none of the variables in $G_{19}$ can be found in the 10 first selected variables. On the contrary, for RFE algorithm, the variables in $G_{19}$ are more selected in the ten first variables (blue in the Figure). In addition, the first selected variable is the variable 18 at each time (green in the Figure). The second and the third selected variables are more frequently the variables 17 and 20 (respectively red and black in the Figure). As a consequence, the rankings of the selected variables by RFE are more stable than NRFE which do not always select the ``true'' variables in the first positions. 
The selection results of our experiment are consistent with those of \citet{maugis11}.

\begin{figure}
	\begin{center}
	\subfloat[RFE]{\includegraphics[scale=0.3, page=1]{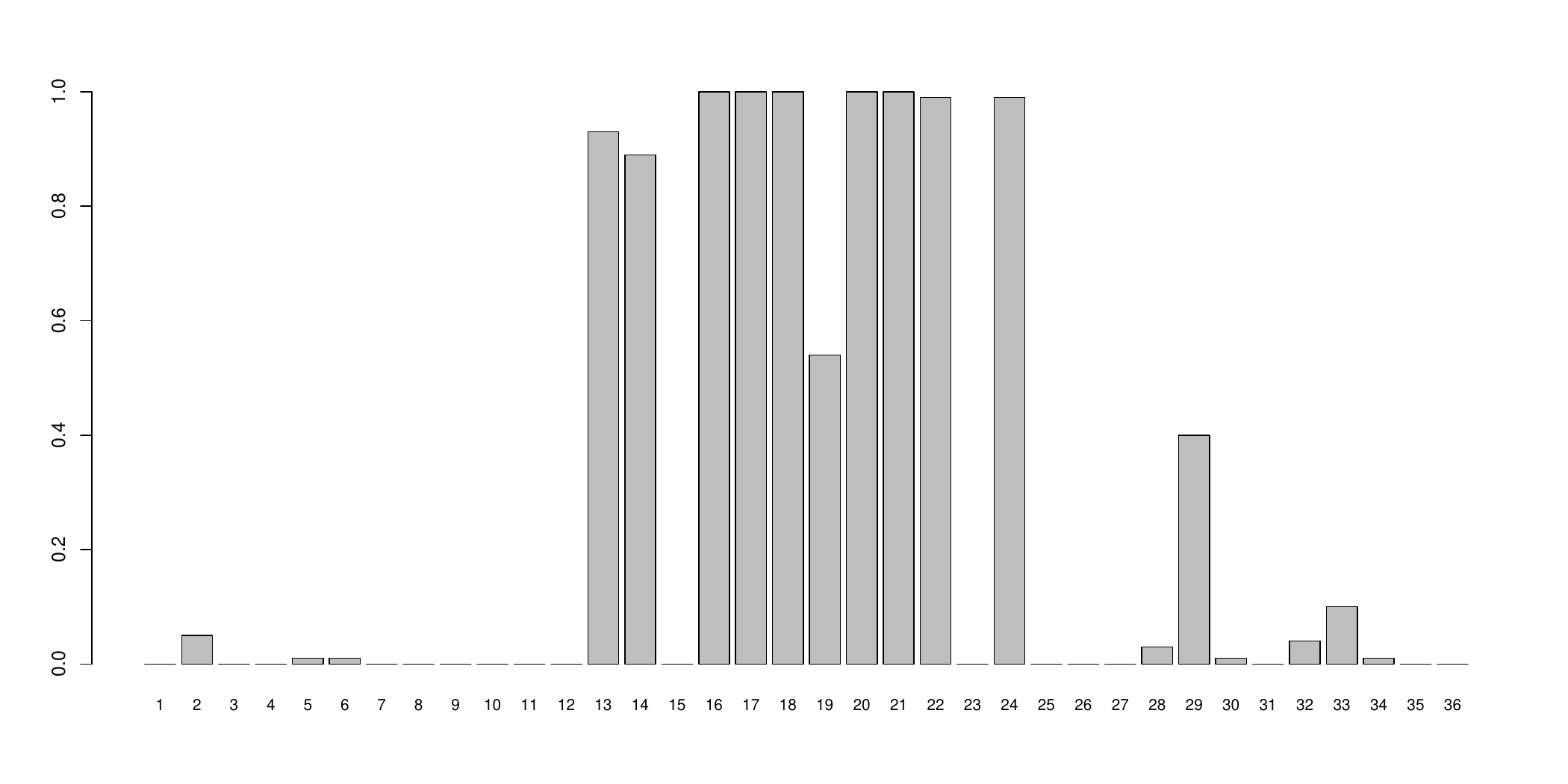}\label{fig:landsatdataFreq1}}
	\hfill
	\subfloat[NRFE]{\includegraphics[scale=0.3, page=1]{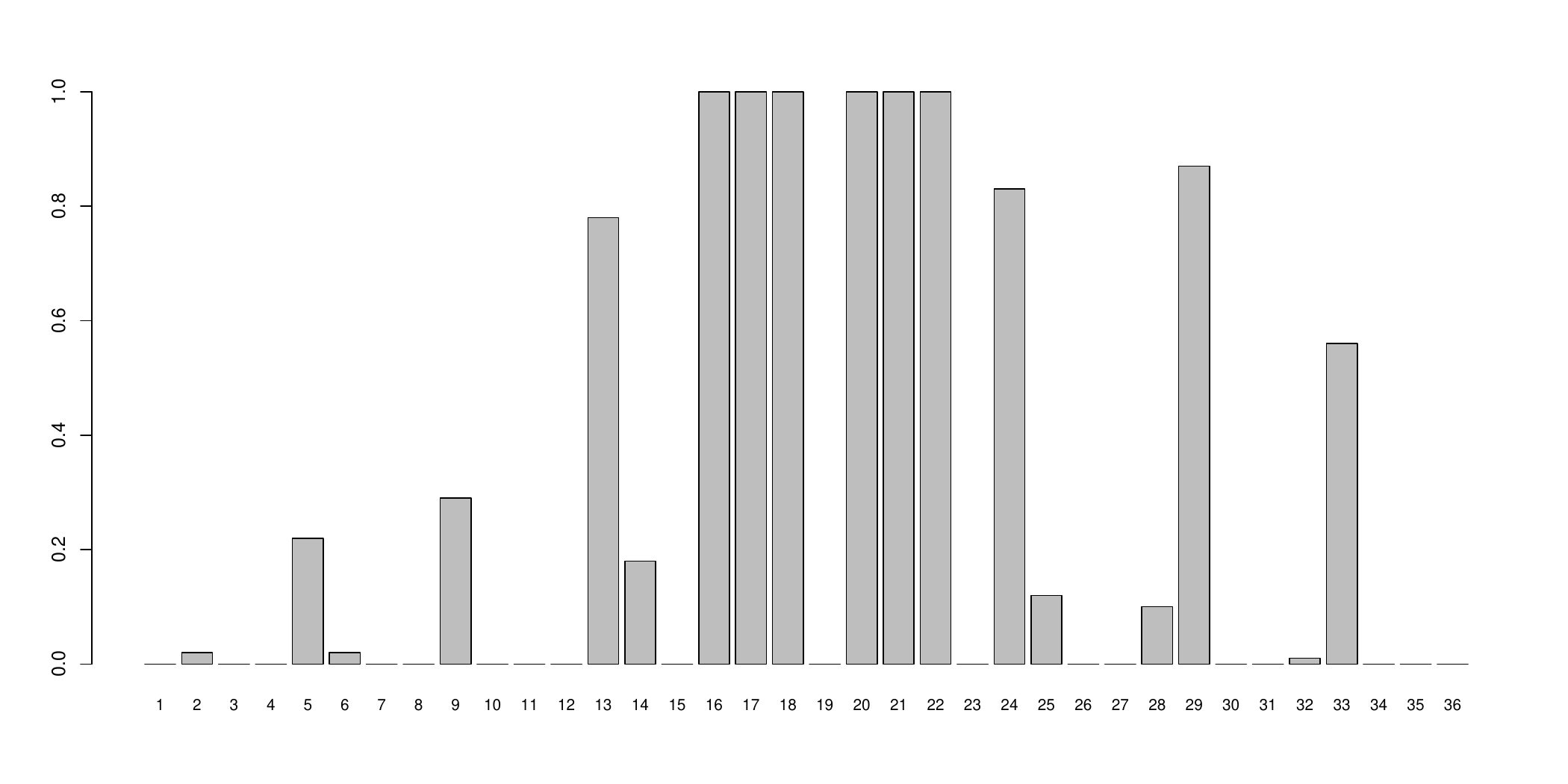}\label{fig:landsatdataFreq2}}
	\caption{Landstat data - Proportion of time each variable is selected in the ten first variables.}
	\label{fig:barplotLandsat}
	\end{center}
\end{figure}

\begin{figure}
	\begin{center}
	\subfloat[The 18 more important variables]{\includegraphics[scale=0.3, page=1]{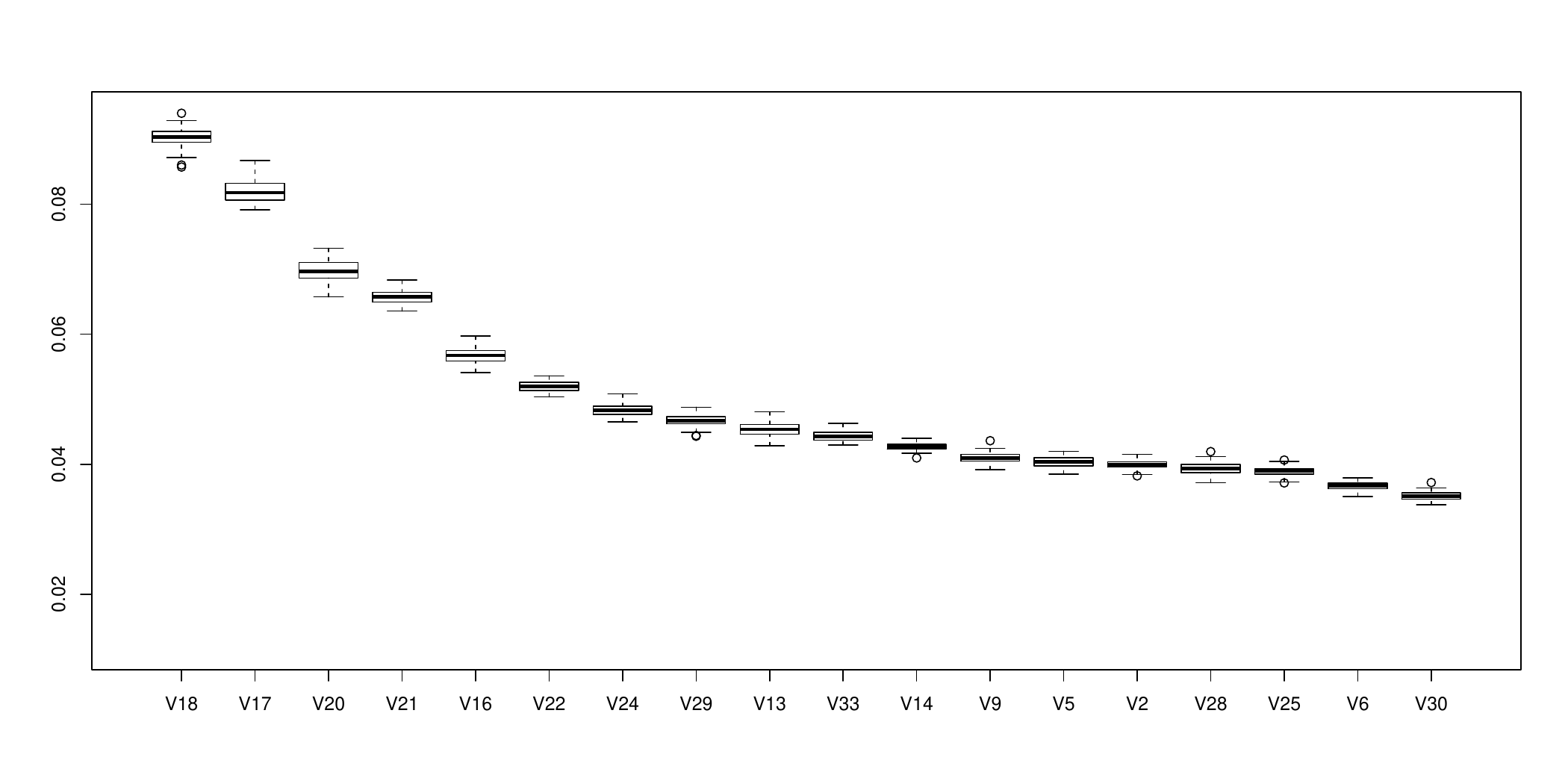}\label{fig:landsatdataBoxplot1}}
	\hfill
	\subfloat[The 18 less important variables]{\includegraphics[scale=0.3, page=2]{satelliteBoxplotVIM.pdf}\label{fig:landsatdataBoxplot2}}
	\caption{Landstat data - Boxplots of the initial permutation importance measures.}
	\label{fig:boxplotLandsat}
	\end{center}
\end{figure}

\begin{figure}
	\begin{center}
	 \includegraphics[scale = 0.7]{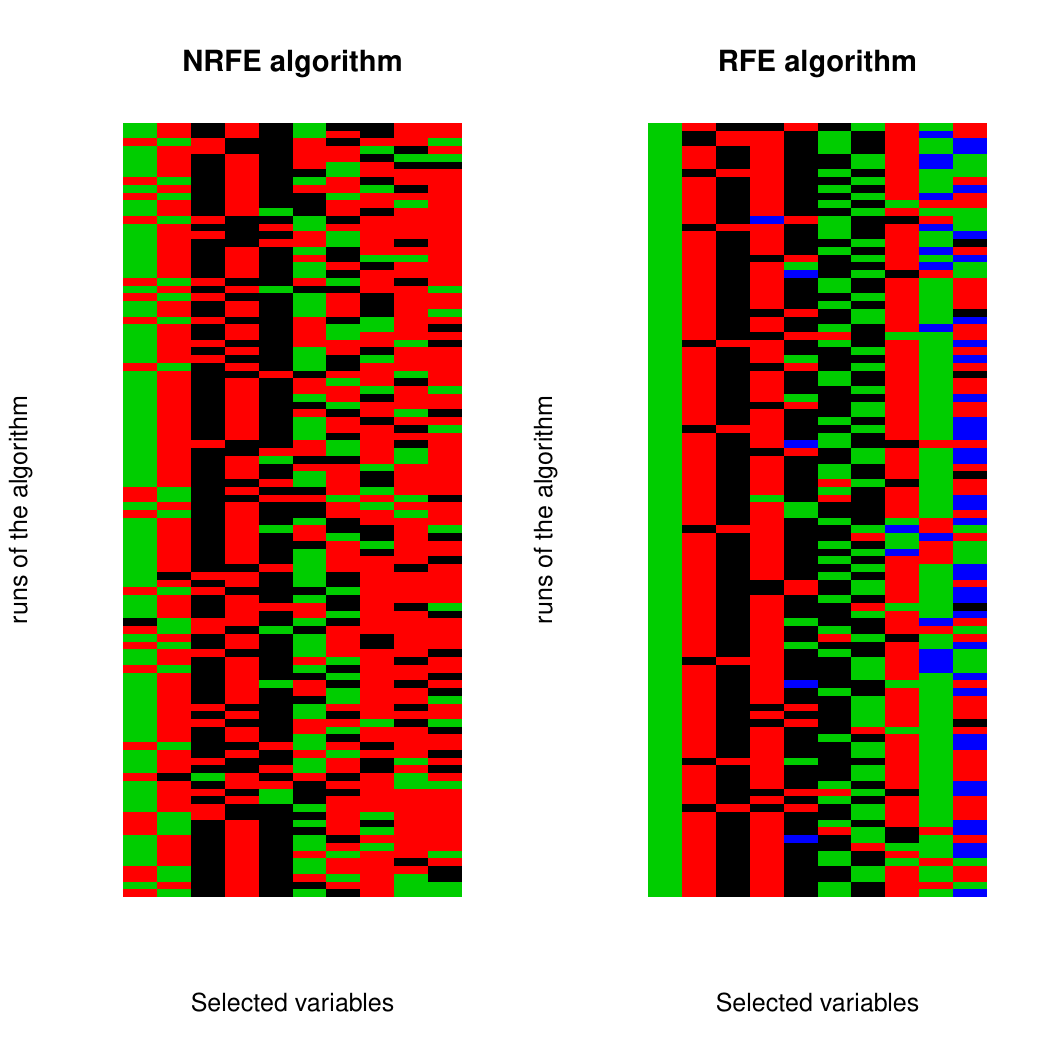}
	\caption{Landstat data - These two diagrams  represent the groups of the ten first variables selected by each algorithm. Each line corresponds to one of the 100 runs of the algorithms. The colors represent the group of the selected variable: red for $G_{17}$, green for $G_{18}$, blue for  $G_{19}$ and black for $G_{20}$.}
	\label{fig:Selec10Landsat}
	\end{center}
\end{figure}

\section{Conclusion}

In this paper, we studied the problem of variable selection using the permutation importance measure from the random forests. Several simulation studies in the literature have shown an effect of the correlations between predictors on this criterion.

We first provided some theoretical insights on the effect of the correlations on the permutation importance measure. Considering an additive regression model, we obtained a very simple expression of this criterion which depends on the correlation between the covariates and on the number of correlated variables. Extending our results to a more general context is a challenging problem, this question should be investigated deeply for improving our knowledge of this widely used criterion. Moreover, the impact of the correlations on other importance measures (see \citet{varImp:vdL06}, \citet{rf:I07}) is a general question of great interest. 

In a second step we focused on variable selection algorithm based on random forests analysis. A recursive and a non recursive approaches have been evaluated through an extensive simulation study on several classification and regression designs. As expected, in presence of correlated predictors, 
the screening given by RFE algorithm  is more efficient than the one proposed by NRFE: the prediction errors is always smaller with recursive strategy when small size models are considered. As a matter of fact RFE reduces the effect of the correlation on the importance measure.

In the RFE algorithm, updating the ranking is especially crucial at the last steps of the algorithm, when most of the irrelevant variables have been eliminated. In future works, the algorithm could be adapted by combining a non recursive strategy at the first steps and a recursive strategy at the end of the algorithm.

\appendix

\section{Proofs}

\subsection{Proof of Proposition \ref{prop:ResAdditif}}
\label{sec:proofResAdditif}

The random variable $X_j'$ and the vector  $\mathbf X_{(j)}$ are defined as in Section \ref{sec:background}:
\begin{align*}
I(X_j) 	&= \E[(Y - f(\mathbf X) + f(\mathbf X) - f(\mathbf X_{(j)}))^2] - \E[(Y-f(\mathbf X))^2]\\
		&= \E[(f(\mathbf X) - f(\mathbf X_{(j)}))^2] + 2 \E \left[ \varepsilon (f(\mathbf X) - f(\mathbf X_{(j)})) \right] \\
		&= \E[(f(\mathbf X) - f(\mathbf X_{(j)}))^2],
\end{align*}
since $ \E [ \varepsilon  f(\mathbf X) ] = \E [   f(\mathbf X) \E[ \varepsilon | \mathbf X ] ]  = 0 $ and $ \E[ \varepsilon  f(\mathbf X_{(j}) ] = \E( \varepsilon)   \E [  f(\mathbf X_{(j}) ] = 0  $. Since the model is additive, we have:
\begin{align*}
I(X_j) 	&= \E[(f_j(X_j) - f_j(X_j'))^2]\\
		&= 2\V[f_j(X_j)],
\end{align*}
as $X_j$ and  $X_j'$ are independent and identically distributed. For the second statement of the proposition, using the fact that $f_j(X_j)$ is centered we have:
\begin{align*}
\C[Y, f_j(X_j)] &= \E \left[  f_j(X_j) \E[Y| \mathbf X] \right] = \E [f_j(X_j) \sum_{k=1}^p f_k(X_k)]\\
				&= \V[f_j(X_j)] + \sum_{k\neq j} \E\left[ f_j(X_j) f_k(X_k) \right] \\		
				&= \frac{I (X_j)} 2  + \sum_{k\neq j} \C\left[ f_j(X_j), f_k(X_k) \right].
\end{align*}

\subsection{Proof of Proposition \ref{prop:Gauss}}
\label{sec:proofGauss}
This proposition is an application of Proposition~\ref{prop:ResAdditif} for a particular distribution. We only show that $ \alpha = C^{-1} \boldsymbol\tau$ in that case.

Since $(\mathbf X, Y)$ is a normal multivariate vector, the conditional distribution of $Y$ over $\mathbf X$ is also normal and the conditional mean $f(\mathbf x) = \E[Y|\mathbf X=\mathbf x]$ is a linear function: $f(\mathbf x) = \sum_{j=1}^p \alpha_j x_j$ (see for instance \citet{rao73}, p. 522). Then, for any $j \in \{1, \dots, p \}$,
	\begin{align*}
	\tau_j & =	\E[X_jY]  \\
 		& = \E[\; X_j \E[Y | \mathbf X] \;] \\
			& = \alpha_1 \E [X_1X_j] + \cdots + \alpha_j  \E [X_j^2] + \cdots + \alpha_p \E [X_p X_j] \\
			& = \alpha_1 c_{1j} + \cdots + \alpha_j c_{jj}+ \cdots + \alpha_p c_{pj}.
	\end{align*}
The vector $\alpha$ is thus solution of the equation $\boldsymbol\tau = C \alpha$ and the expected result is proven since the covariance matrix $C$ is invertible.

\subsection{Proof of Proposition~\ref{lem:invC}}
\label{sec:proofinvC}
The correlation matrix $C$ is assumed to have the form $C = (1-c) I_p + c \indd \indd^t$. We show that the invert of $C$ can be decomposed in the same way. Let $M = a I_p + b \indd \indd^t$ where $a$ and $b$ are real numbers to be chosen later. Then

\begin{eqnarray*}
C M	&=& \big( (1-c)I_p + c \indd \indd^t \big) \big( a I_p + b \indd \indd^t \big) \\
			&=& a (1-c) I_p + b (1-c) \indd \indd^t + a c \indd \indd^t + b c \indd \indd^t \indd \indd^t \\
			&=& a (1-c) I_p + (b (1-c) + ac + pbc )\indd \indd^t,
\end{eqnarray*}
since $\indd^t \indd = p$. Thus, $C M = I_d$ if and only if
\[
\left\{
	\begin{array}{l}
	a (1-c) = 1 \\
	b (1-c) + ac + pbc = 0,
	\end{array}
\right.
\]
which is equivalent to
\[
\left\{
	\begin{array}{l}
	a = \dfrac{1}{(1-c)} \\
	b = \dfrac{- c}{(1-c)(1-c+pc)}.
	\end{array}
\right.
\]
Consequently,  $M^{-1}_{jk} = C^{-1}_{jk} = b$ if $j \neq k$ and $M^{-1}_{jk} = C^{-1}_{jj} = a+b$. Finally we find that for any $j \in \{1\dots p\}$:
\begin{eqnarray*}
[C^{-1} \boldsymbol\tau]_j 	&=& \tau_0 (a+b) + \tau_0 b (p-1)   \\
						&=& \tau_0 ( a + pb )  \\
						&=& \tau_0 \bigg( \dfrac{1}{(1-c)} -  \dfrac{pc}{(1-c)(1-c+pc)} \bigg) \\
						&=& \dfrac{\tau_0}{1 - c + pc}.
\end{eqnarray*}
The second point derives from Proposition~\ref{prop:Gauss}.

\section*{Acknowledgements}
The authors would like to thank G\'erard Biau for helpful discussions and Cathy Maugis for pointing us the Landsat Satellite dataset. The authors also thank the two anonymous referees for their many helpful comments and valuable suggestions.

\bibliographystyle{plainnat}
\bibliography{biblio}

\end{document}